\documentclass[aps,prx,reprint,groupedaddress,twoside,showpacs]{revtex4-2}

\usepackage{amsmath,amssymb}
\usepackage{amsthm}
\usepackage{hyperref} % create hyperlinks
\usepackage{acronym}  % make an acronym
\usepackage{color}  % make colors
\usepackage{amsfonts}
\usepackage{mathrsfs}
\usepackage{graphicx}
\usepackage[on]{psfrag}
\usepackage{array}
\usepackage{cases}
\usepackage{tabularx}
\usepackage{url}
\usepackage{soul}
\usepackage{theoremref}
\usepackage{thmtools}
\usepackage{lettrine}
\usepackage{mathtools}
\usepackage{dsfont}
\usepackage[table]{xcolor}
\usepackage{multirow}
\usepackage{dashbox}
\usepackage{algorithm}% http://ctan.org/pkg/algorithms
\usepackage{algpseudocode}% http://ctan.org/pkg/algorithmicx
\usepackage{float}
\usepackage{enumitem}
\usepackage{hhline}
\usepackage{dashrule}
\usepackage{ulem}
\usepackage{empheq}

\newcommand\reddashuline{\bgroup\markoverwith
{\textcolor{red}{
\hdashrule[-0.5ex]{1pt}{0.8pt}{0.5pt}}}\ULon}
\newcommand\blueuline{\bgroup\markoverwith
{\textcolor{blue}{\rule[-0.5ex]{2pt}{0.8pt}}}\ULon}
\usepackage{tikz}

\usepackage{xpatch}

\makeatletter
\let\saved@includegraphics\includegraphics
\AtBeginDocument{\let\includegraphics\saved@includegraphics}
\renewenvironment*{figure}{\@float{figure}}{\end@float}
\makeatother

\makeatletter
\renewcommand{\ALG@name}{Protocol}
\makeatother

\usepackage{relsize}%
\usepackage{WinsNotation}

\makeatletter
\renewcommand*\env@matrix[1][\arraystretch]{%
  \edef\arraystretch{#1}%
  \hskip -\arraycolsep
  \let\@ifnextchar\new@ifnextchar
  \array{*\c@MaxMatrixCols c}}
\makeatother

\makeatletter
\newcommand{\multiline}[1]{%
  \begin{tabularx}{\dimexpr\linewidth-\ALG@thistlm}[t]{@{}X@{}}
    #1
  \end{tabularx}
}
\makeatother

\newcommand{\SM}{\hspace{-0.15mm}\Set M}
\newcommand{\tSM}{\hspace{-0.15mm}\tilde{\Set M}}
\newcommand{\TVj}{\hspace{0.4mm}\tilde{\hspace{-0.4mm}\V{j}\hspace{0.4mm}}\hspace{-0.4mm}}

\makeatletter
\renewcommand*\env@matrix[1][\arraystretch]{%
  \edef\arraystretch{#1}%
 \hskip -\arraycolsep
  \let\@ifnextchar\new@ifnextchar
  \array{*\c@MaxMatrixCols c}}
\makeatother

\makeatletter
    \def\CT@@do@color{%
      \global\let\CT@do@color\relax
            \@tempdima\wd\z@
            \advance\@tempdima\@tempdimb
            \advance\@tempdima\@tempdimc
    \advance\@tempdimb\tabcolsep
    \advance\@tempdimc\tabcolsep
    \advance\@tempdima2\tabcolsep
            \kern-\@tempdimb
            \leaders\vrule
                    \hskip\@tempdima\@plus  1fill
            \kern-\@tempdimc
            \hskip-\wd\z@ \@plus -1fill }
\makeatother

\makeatletter
\def\hlinew#1{%
  \noalign{\ifnum0=`}\fi\hrule \@height #1 \futurelet
   \reserved@a\@xhline}
\newcolumntype{"}{@{\hskip\tabcolsep\vrule width 0.8pt\hskip\tabcolsep}}   
\makeatother

\newcolumntype{L}[1]{>{\raggedright\let\newline\\\arraybackslash\hspace{0pt}}m{#1}}
\newcolumntype{C}[1]{>{\centering\let\newline\\\arraybackslash\hspace{0pt}}m{#1}}
\newcolumntype{R}[1]{>{\raggedleft\let\newline\\\arraybackslash\hspace{0pt}}m{#1}}

\definecolor{dblue}{rgb}{0,0,0.75}
\definecolor{dgreen}{rgb}{0,0.7,0}
\definecolor{lightblue}{rgb}{0.83,0.91,1.0}
\definecolor{golden}{rgb}{0.95,0.75,0.2}

\newcommand*{\QEDB}{\hfill\ensuremath{\square}}%

\DeclareMathOperator*{\Motimes}{\text{\raisebox{0.2ex}{\scalebox{0.6}{$\bigotimes$}}}}

\newcommand{\Ps}{\mathscr{P}}

\declaretheoremstyle[
spaceabove=8pt, spacebelow=8pt,
headfont=\normalfont\bfseries,
notefont=\normalfont, notebraces={(}{)},
headpunct={:},
bodyfont=\normalfont,
postheadspace=0.5em,
]{mynote}
\declaretheoremstyle[
spaceabove=4pt, spacebelow=4pt,
headfont=\normalfont,
notefont=\normalfont, notebraces={(}{)},
headpunct={:},
bodyfont=\normalfont,
postheadspace=0.5em,
]{mynote2}

\declaretheorem[style=mynote]{Remark}
\declaretheorem[style=mynote,name=Theorem]{Thm}
\declaretheorem[style=mynote,name=Lemma]{Lem}

\declaretheorem[style=mynote,name=Proposition]{Prop}

\declaretheorem[style=mynote,name=Definition]{Def}
\declaretheorem[style=mynote2,name=$\Ps$-\hspace{-1.2mm}]{Problem}

\newlength{\wdth}
\newcommand{\strike}[1]{\settowidth{\wdth}{#1}\rlap{\rule[.5ex]{\wdth-0.3ex}{.4pt}}#1}
\newcommand{\sstrike}[1]{\settowidth{\wdth}{#1}\rlap{\rule[.35ex]{\wdth-0.45ex}{.4pt}}#1}

\usepackage[textsize=scriptsize]{todonotes}

\newcommand{\paperTitle}{Error minimization for fidelity estimation of GHZ states with arbitrary noise}

\begin{document}
% user normal Italic font rather than underline for \em command
\normalem
%\markboth{\small PLEASE DO NOT DISTRIBUTE WITHOUT THE WRITTEN CONSENT OF THE AUTHORS \version}{\small Ruan, Dai, and Win: \paperTitleMarkboth}
%\markboth{Please do not distribute the manuscript without the written consent of the authors. \version}{Ruan, Dai, and Win: \paperTitleMarkboth}

%\thispagestyle{empty}
{\color{white}
\fontsize{0pt}{0pt}\selectfont
\begin{acronym}
\acro{LOCC}{local operations and classical communication}\vspace{-5.5mm}
\acro{QED}{quantum entanglement distillation}\vspace{-5.5mm}
\acro{PoE}{probability of entanglement}\vspace{-5.5mm}
\acro{w.r.t.}{with respect to}\vspace{-5.5mm}
\acro{i.i.d.}{independent and identically distributed}\vspace{-5.5mm}
\acro{pmf}{probability mass function}\vspace{-5.5mm}
\acro{pdf}{probability density function}\vspace{-5.5mm}
\acro{cdf}{cumulative distribution function}\vspace{-5.5mm}
\acro{iff}{if and only if}\vspace{-5.5mm}
\acro{POVM}{positive operator-valued measure}\vspace{-5.5mm}
\acro{QBER}{quantum bit error rate}\vspace{-5.5mm}
\acro{MSE}{mean square error}\vspace{-5.5mm}
\acro{MVUE}{minimum-variance unbiased estimator }\vspace{-5.5mm}
\acro{GHZ}{Greenberger-Horne-Zeilinger}\vspace{-5.5mm}
\acro{EPR}{Einstein-Podolsky-Rosen}\vspace{-5.5mm}
\acro{QST}{quantum state tomography}\vspace{-5.5mm}
\acro{QKD}{quantum key distribution}\vspace{-5.5mm}
\acro{QED}{Quantum error detection}\vspace{-5.5mm}
%\acrodef{CI}{credible interval}\vspace{-5.5mm}
\acro{NV}{Nitrogen-vacancy}
\end{acronym}}

%\newpage
\setcounter{page}{1}
%\title{Adaptive Recurrence Quantum Distillation for Two-Kraus-Operator Channels}
%\title{\paperTitle}
\title{\paperTitle}

\author{Liangzhong Ruan}
%\thanks{}
\affiliation{
School of Cyber Science and Engineering, Xi'an Jiaotong University, China
}


%\date{\today}
%\date{\version}

\begin{abstract}
Fidelity estimation is a crucial component for the quality control of entanglement distribution networks.
This work studies a scenario in which multiple nodes share noisy \ac{GHZ} states. 
Due to the collapsing nature of quantum measurements, the nodes randomly sample a subset of noisy \ac{GHZ} states for measurement and then estimate the average fidelity of the unsampled states conditioned on the measurement outcome.
By developing a fidelity-preserving diagonalization operation, analyzing the Bloch representation of \ac{GHZ} states, and maximizing the Fisher information, the proposed estimation protocol achieves the minimum mean squared estimation error in a challenging scenario characterized by arbitrary noise and the absence of prior information. 
Additionally, this protocol is implementation-friendly as it only uses local Pauli operators according to a predefined sequence.
Numerical studies demonstrate that, compared to existing fidelity estimation protocols, the proposed protocol reduces estimation errors in both scenarios involving independent and identically distributed (i.i.d.) noise and correlated noise.

\end{abstract}

%\keywords{}

\maketitle

\acresetall             % reset the acronyms

\section{Introduction}
Entanglement is a primary source of quantum advantage in the emerging quantum networks, such as the quantum Internet \cite{WehElkHan:J18,Panetal:J19} and quantum computing networks \cite{Pre:J18,Havetal:J19}.
In particular, multi-party entanglement plays a critical role in various applications, such as multi-party quantum communication \cite{KuzVasSanDurMus:J19}, computation \cite{SonGouWen:J20}, and clock synchronization \cite{BenExm:J11}.
In this context, evaluation of the quality of multi-party entanglement is critical for the development of medium to large scale quantum networks.
\ac{QED} \cite{ZhoWanMarCleKorWan:J14,Cordetal:J15,Linetal:J17} and \ac{QST} \cite{DArLoP:J01,HusHou:J12,LuLiuZha:J15}
can both in principle achieve this task.
However, to mitigate the effects of noise, 
\ac{QED} must ensure sufficient information redundancy by employing sophisticated quantum circuits, and
to obtain complete information about the state, \ac{QST} requires the number of measurements to grow exponentially \ac{w.r.t.}  the dimension of the state \cite{AleBor:J18}.
The advancement of quantum networks necessitates the development of resource-efficient methods for evaluating the quality of multi-party entanglement.

Fidelity indicates the quality of quantum states \cite{BriDurCirZol:98,WesLidFonGyu:J10,Basetal:J11,ArrLazHelBal:J14} and can be estimated using separable quantum operations and classical post-processing.
Separable quantum operations do not require nodes to have quantum memories, making the estimation protocols feasible with state-of-the-art quantum technologies \cite{YinCaoPanetal:J17,Humetal:J18,Pom-etal:J21,LiuTianGuetal:J20}.
Consequently, fidelity estimation emerges as a promising method for evaluating the quality of entanglement.
Fidelity estimation protocols have been developed for various entangled states \cite{TokYamKoaImo:J06,SomChiBer:J06,GuhLuGaoPan:J07, FlaLiu:J11, ZhuHay:J19}.
However, to estimate the fidelity of multi-party entangled states, several key challenges remain to be addressed.

\begin{itemize}[leftmargin=*]

\item{\em Limit the detrimental effect of arbitrary noise:}
Quantum networks frequently experience heterogeneous and correlated noise.
This type of noise reduces the amount of information that can be derived from measurements, consequently resulting in excessive estimation error in various quantum protocols \cite{TomLimGisRen:J12,PfiRolManTomWeh:J18}.
The reduction of the excessive errors caused by arbitrary noise, specifically in the context of fidelity estimation for multi-parity entangled states, represents a promising area of research.

\item{\em Enhance the efficiency of separable operations:}
Separable quantum operations are advantageous for implementation as they do not necessitate quantum memories.
However, as has been noted in the context \ac{QST}, in contrast to joint operations, separable operations can result in significant measurement efficiency loss \cite{BagBalGilRom:J06}.
Addressing the efficiency loss associated with separable operations in the fidelity estimation of multi-parity entangled states continues to present a notable challenge.
\end{itemize}

The author's prior work \cite{Ruan:J23} addresses the above chalaforementioned challenges in the context of two-party entangled states.
In this paper, we focus on the fidelity estimation of \ac{GHZ} states, 
which are a set of maximally entangled states shared by an arbitrary number of parties.
These states have numerous applications in quantum networks, including secure communication \cite{GaoQinWenZhu:J10,QinDai:J17}, error correction \cite{Den:J11,KelFowetal:J15}, and distributed computation \cite{GotChu:J99,HonPan:J06}.
Given that measurements collapse quantum states, we consider a network in which nodes randomly sample a subset of the shared noisy \ac{GHZ} states to measure and estimate the average fidelity of unsampled noisy states conditioned on the measurement outcome.

Given the limitation of having only separable measurements, we utilize the Bloch representation to decompose the \ac{GHZ} states into the sum of local operators.
By employing this decomposition alongside Fisher information \cite{LyMarVerGraWag:J17}, we propose a measurement operation that minimizes the mean squared estimation error of fidelity.
The proposed operation is practical for implementation as it involves only local Pauli operators.

The remainder of the paper is organized as follows.
Section~\ref{sec:formulation} formulates the fidelity estimation problem.
Section~\ref{sec:transform} and~\ref{sec:construction} present the theoretical foundation and detail the construction of the proposed protocol.
Section~\ref{sec:sim} provides the numerical studies.
Finally, Section~\ref{sec:summary} offers the concluding remarks.

Random variables and their realizations are denoted in upper- and lower-case letters, respectively, e.g., $F$ and $f$.
%Appendix~\ref{sec:symbol} presents detailed lists of math notations, frequently used symbols, and key statistical notions.

\section{System Setup and Problem Formulation}

\label{sec:formulation}

Consider a quantum network consisting of $L(\geq 2)$ nodes that share multiple copies of noisy entangled states.
The nodes are indexed by $l \in {\Set L}= \{1,2,\ldots,L\}$. 
The target shared state is one of the $L$-qubit  \ac{GHZ} states, which are defined as
\begin{align}
|G^{\pm}_{\V{j}}\rangle = \frac{1}{\sqrt{2}}\big(|\V{j}\rangle \pm  | \TVj \rangle\big),\label{eqn:gGHZ}
\end{align}
where 
\begin{align}
\V{j} = [j_1,j_2,\ldots,j_L]\in {\Set J} = \{0\}\! \frown \!\{0,1\}^{L-1}\label{eqn:Vj}
\end{align}
 is an $L$-bit binary string, in which $\frown $ denotes the concatenation of two strings, and $\TVj$ is the bitwise complement of $\V{j}$.
 
Denote the density matrix of $|G^{\pm}_{\V{j}}\rangle$ by 
$ \M{G}^{\pm}_{\V{j}} $, denote the set of all  $L$-qubit \ac{GHZ} states by $\Set{G}_L$, and
denote the target state by $|G^{ s }_{\V{t}}\rangle$, where
$ s  \in\{+,-\}$, $\V{t}\in\Set J$.

The nodes share $N$ noisy $L$-qubit \ac{GHZ} states and have no prior information of the noise.
To assess the quality of entanglement, the nodes randomly sample $M(<N)$ noisy  \ac{GHZ} states for measurement.
Since the measured states are no longer entangled, the nodes estimate the fidelity of the unsampled states conditioned on the measurement outcomes. 
The set of all noisy \ac{GHZ} states and that of the sampled states are denoted by $\Set N$, with $|\Set N| = N$, and $\Set{M}$, with $|\SM| = M$, respectively.
Here, $\Set M$ is drawn from all $M$-subsets of $\Set N$ with equal probability

To be implementable on nodes without local quantum memory, the estiamtion protocol employs operators that are separable across all qubits.  
These operators are denoted by 
\begin{align}
\M{M}_{r}= \sum_{k}\otimes^{n\in \Set M} (\otimes^{l\in \Set L}\M{M}^{(l)}_{r,n,k}),\label{eqn:MXsep}
\end{align} 
where the \ac{GHZ} state index $n\in \Set M$, the node index $l\in\Set{L}=\{1,2,\ldots,L\}$, and operator index $r\in \Set R$. 
Here $\Set R$ represents the set of all possible measurement outcomes.
As \ac{POVM} operators, 
\begin{align}
\M{M}^{\mathrm{(X)}}_{r,n,k}\succcurlyeq \M{0}, \; \sum_{r}\M{M}_{r} =\mathbb{I}_{2^{LM}}, 
\label{eqn:MPOVM}
\end{align}
where $\succcurlyeq$ denotes matrix inequality.
The set of all measurement operators is denoted by
\begin{align}
\Set{O}= \{\M{M}_{r},\; r\in \Set R \},\label{eqn:measurement}
\end{align}

Denote $\V{\rho}_{\mathrm{all}}$ as the density matrix of all \ac{GHZ} states before being measured, and
denote the $\V{\rho}_{\mathrm{all}}^{(r)}$ as the density matrix of the \ac{GHZ} states conditioned on the measurement outcome $r$.
In this case, the density matrix of the $n$-th \ac{GHZ} state is
\begin{align}
\V{\rho}^{(r)}_n = \mathrm{Tr}_{i\in{\Set N}\backslash \{n\}} \V{\rho}_{\mathrm{all}}^{(r)}.
\end{align}
Based on the measurement outcome, the estimator $\Set D$ is used to estimate the average fidelity of the unsampled \ac{GHZ} states $\bar{f}$, i.e.,
\begin{align}
\check{f} = {\Set D}(r).\label{eqn:estimator}
\end{align}
Because each 
$M$-subset is selected with probability ${N\choose M}^{-1}$, the mean squared estimation error is given by 
\begin{align}{N\choose M}^{-1} \sum_{\Set M}\mathbb{E}_{R}\big[(\check{F}-\bar{F})^2\big],
\end{align}
where $\check{F}$, $\bar{F}$, and $R$ are the random variable form of the estimated fidelity $\check{f}$, the average fidelity $\bar{f}$, and the measurement outcome $r$, respectively.

To optimize the efficiency of the measurement without prior information of $\V{\rho}_{\mathrm{all}}$, we design the measurement operators to achieve a low estimation error for all types of noise.
Specifically, the squared error of the estimated fidelity is minimized for the worst-case state $\V{\rho}_{\mathrm{all}}$,
subject to the constraint that the estimation is unbiased.
This problem is formulated as
\begin{Problem}[Error minimization with arbitary noise] \label{prob:1-m}
\begin{subequations}
\begin{align}
 \underset{\Set O, \Set D}{\mathrm{minimize}}\;\;&\max_{\V{\rho}_{\mathrm{all}}\in{\Set S}_{\mathrm{arb}}}
\sum_{\Set M}\mathbb{E}_{R}\big[(\check{F}-\bar{F})^2\big] 
\label{eqn:minvar_0}\\
\mathrm{subject\; to}\;\;
&\sum_{\SM}\mathbb{E}_{R}\big[\check{F}-\bar{F}\big]=0 ,\;\; \forall \V{\rho}_{\mathrm{all}}\in{{\Set S}_{\mathrm{arb}}} \label{eqn:unbiased_0}
\end{align}
\end{subequations}
where ${\Set S}_{\mathrm{arb}}$ denotes the set of arbitrary $N\times L$ qubit states, and
$\Set O$, $\Set D$ are defined in \eqref{eqn:measurement}, \eqref{eqn:estimator} respectively.
The factor ${N\choose M}^{-1}$ it is omitted in \eqref{eqn:minvar_0} because it does not affect the optimization of the esitmation protocol $\{\Set O, \Set D\}$.
\end{Problem}

The following two sections elaborate the procedures for solving $\Ps$-1, which involve two main components: problem transformation and operation construction.

\section{Problem Transformation}
\label{sec:transform}
This section demonstrates that $\Ps$-\ref{prob:1-m} can be transformed into an equivalent problem with independent noise.
The transformation process involves three steps.
\subsection{Sufficiency of considering classical correlated noise}
This step simplifies $\Ps$-\ref{prob:1-m} to an equivalent problem with classical correlated noise, specifically,
\begin{Problem}[Error minimization with separable noise] \label{prob:sp-m}
\begin{subequations}
\begin{align}
 \underset{\Set O, \Set D}{\mathrm{minimize}}\;\;&\max_{\V{\rho}_{\mathrm{all}}\in{\Set S}_{\mathrm{sp}}}\sum_{\SM}\mathbb{E}_{R}\big[(\check{F}-\bar{F})^2\big] 
\label{eqn:minvar_sep}\\
\text{subject to}\;\;
&\sum_{\SM}\mathbb{E}_{R}\big[\check{F}-\bar{F}\big]=0 ,\;\; \forall \V{\rho}_{\mathrm{all}}\in{\Set S}_{\mathrm{sp}}\label{eqn:unbiased_sep}
\end{align}
\end{subequations}
where
${\Set S}_{\mathrm{sp}}$ denotes the set of states $\V{\rho}_{\mathrm{all}}$ that are separable across different $L$-qubit groups.
\end{Problem}

We first construct an operation $\Set T$ that removes the entanglement among different \ac{GHZ} states without altering the fidelity.
\begin{Def}[Probabilistic multirotation]\thlabel{def:T}
 $\M{T}^{(\V{k})}$ is an $L$-qubit operator that rotates each qubit by $180^\circ$ around the $x$- or $y$-axis according to the string $\V{k} \in {0,1}^L$, i.e.,
\begin{align}
\M{T}^{(\V{k})} =\V{\sigma}_{xy}^{(\V{k})} = \Motimes_{l\in\Set L} \V{\sigma}_{xy}^{(k_l)},\label{eqn:Txy}
\end{align}
where $k_l$ is the $l$-th element of $\V{k}$, and%$|\V{k}| = \sum_{l \in \Set L} \V{k}[l]$ is the Hamming weight of a bit string and
\begin{equation}
\V{\sigma}_{xy}^{(k_l)} := \begin{cases}
\V{\sigma}_y, & \mbox{if } k_l = 1, \\
\V{\sigma}_x, & \mbox{if } k_l = 0, \\
\end{cases}\label{eqn:sigmaxy}
\end{equation}
in which $\V{\sigma}_x$ and $\V{\sigma}_y$ are the Pauli X and Y matrices, respectively.

Probabilistic operation $\Set{T}^{(\V{k})}$ either performs no action or applies the rotation $\M{T}^{(\V{k})}$, each with a probability of 0.5, i.e.,
the two Kraus operators of $\Set{T}^{(\V{k})}$ are
\begin{align}
\frac{1}{\sqrt{2}}\mathbb{I}_{2^L} \;\mbox{ and }\;\frac{1}{\sqrt{2}}\M{T}^{(\V{k})}.\label{eqn:Tk}
\end{align}
Define a set of $\Set{T}^{(\V{k})}$ as
\begin{align}
\Set{T}_{\mathrm{all}} = \big\{\Set{T}^{(\V{k})}:|\V{k}|\in\{0,2\}\big\}.
\label{eqn:T_all}
\end{align}
Operation $\Set T$ sequentially performs all $\Set{T}^{(\V{k})}\in \Set{T}_{\mathrm{all}} $  on each \ac{GHZ} state.
~\QEDB
\end{Def}

Using operation $\Set T $, states in ${\Set S }_{\mathrm{arb}}$ can be transformed into those in ${\Set S }_{\mathrm{sp}}$ and then estimated using the optimal solution of $\Ps$-\ref{prob:sp-m}. 
This ensures that the minimum estimation error for states in ${\Set S }_{\mathrm{arb}}$ is no higher than that for states in ${\Set S }_{\mathrm{sp}}$.

\thref{lem:sp2arb} provides a formal proof of the analysis above.

\begin{Lem}[Equivalence of $\Ps$-\ref{prob:1-m} and $\Ps$-\ref{prob:sp-m}] \thlabel{lem:sp2arb}
If a measurement operation $\Set O^*$ is optimal in $\Ps$-\ref{prob:sp-m}, 
the composite operation $\hat{\Set O}^* = \Set O^* \circ \Set T$ is optimal in $\Ps$-\ref{prob:1-m}.
%Moreover... (We may add a sentence here claiming that considering arbitrary state does not increase the min square error)
\end{Lem}
\begin{proof}
Please refer to Appendix~\ref{pf_lem:sp2arb}.
\end{proof}

\subsection{Sufficiency of considering independent noise}
This step demonstrates that to solve $\Ps$-\ref{prob:sp-m}, it is sufficient to consider the special case with independent noise, i.e.,
\begin{Problem}[Error minimization for independent states] \label{prob:id-m}
\begin{subequations}
\begin{align}
 \underset{\Set O, \Set D}{\mathrm{minimize}}\;\;&
 \max_{\V{\rho}_{\mathrm{all}}\in{\Set S}_{\mathrm{id}}}\sum_{\SM}\mathbb{E}_{R}\big[(\check{F}-\bar{f}\,)^2\big] 
\label{eqn:minvar_id}\\
\text{subject to}\;\;
&\sum_{\SM}\mathbb{E}_{R}\big[\check{F}-\bar{f}\,\big]=0 ,\;\; \forall \V{\rho}_{\mathrm{all}}\in{\Set S}_{\mathrm{id}}\label{eqn:unbiased_id}
\end{align}\label{eqn:P3}
\end{subequations}
\end{Problem}
where
${\Set S}_{\mathrm{id}}$ denotes the set of states $\V{\rho}_{\mathrm{all}}$ in which the states
of all $N$ noisy \ac{GHZ} states are independent, i.e., $\V{\rho}_{\mathrm{all}}$ can be expressed as
\begin{align}
\V{\rho}_{\mathrm{all}} = \otimes^{n\in\Set N} \V{\rho}_{n}.\label{eqn:rhoall-id}
\end{align}
In the case of \eqref{eqn:rhoall-id}, the average fidelity of the unsampled states, $\bar{f}$, is deterministic for each sample set $\SM$.
Therefore, lowercase letter $\bar{f}$ is used in \eqref{eqn:P3}.

Correlated noise leads to atypical measurement outcomes, resulting in higher estimation errors \cite{PfiRolManTomWeh:J18}.
To bound the estimation error with correlated noise by that with independent noise, first note that any separable state $\V{\rho}_{\mathrm{all}}\in \Set{S}_{\mathrm{sp}}$ can be expressed as
\begin{align}
\V{\rho}_{\mathrm{all}} = \sum_{k\in \Set K}p_k \otimes^{n\in\Set N} \V{\rho}_{k,n}, 
\label{eqn:Rall-sep}
\end{align}
where $p_k$ is the probability of the ensemble being in case $k$,
\begin{align}
p_k\geq 0, \;  \sum_{k\in \Set K}p_k=1,
\end{align}
and
$\V{\rho}_{k,n}$ is the density matrix of the $n$-th noisy \ac{GHZ} state in case $k$.
Denote 
\begin{align}
\V{\rho}^{(k)}_{\mathrm{all}} = \otimes^{n\in\Set N} \V{\rho}_{k,n}.
\end{align}
% and denote the fidelity of $\V{\rho}_{k,n}$ by $f_{k,n}$.
Then $\V{\rho}^{(k)}_{\mathrm{all}} \in \Set{S}_{\mathrm{id}}$, $\forall k \in \Set K$. 

According to Bayes' theorem.
\begin{align}
\mathrm{Pr}(\V{\rho}^{(k)}_{\mathrm{all}}|  r) &=\frac{p_k}{\mathrm{Pr}(r)}\mathrm{Pr}(r|\V{\rho}^{(k)}_{\mathrm{all}}).\label{eqn:Bayes}
\end{align}
%which shows that the measurement outcome $r$  postselects components of $\V{\rho}_{\mathrm{all}}$.
Given the ensemble $\V{\rho}_{\mathrm{all}}$ and the measurement operation $\Set O$,
$p_k$  and $\mathrm{Pr}(r)$ are determined.
In this context, \eqref{eqn:Bayes} indicates that if a measurement outcome $r$ is atypical for component $\V{\rho}^{(k)}_{\mathrm{all}}$, i.e.,
$\mathrm{Pr}(r|\V{\rho}^{(k)}_{\mathrm{all}})$ is close to $0$, 
then the likelihood of this component evaluated at point $r$, i.e.,
$\mathrm{Pr}(\V{\rho}^{(k)}_{\mathrm{all}}|  r)$ is also close to $0$.
Because the average fidelity $\bar{f}$ is estimated conditioned on the measurement outcome $r$, 
components $\V{\rho}^{(k)}_{\mathrm{all}}$ with likelihood close to $0$ have limited effect on the estimation error.
This result provides a method to bound the estim ation error with correlated noise.

\thref{lem:id2sp} provides a formal proof of the analysis above.

\begin{Lem}[Equivalence of $\Ps$-\ref{prob:sp-m} and $\Ps$-\ref{prob:id-m}] \thlabel{lem:id2sp}
If a measurement operation $\Set O^*$ is optimal in $\Ps$-\ref{prob:id-m}, 
it is also optimal in $\Ps$-\ref{prob:sp-m}.
\end{Lem}
\begin{proof}
Please refer to Appendix~\ref{pf_lem:id2sp}
\end{proof}

\subsection{Sufficiency of considering the sampled \ac{GHZ} states}
This step demonstrates that to solve $\Ps$-\ref{prob:id-m},
it is sufficient to minimize the estimation error of the sampled \ac{GHZ} states, i.e.,
\begin{Problem}[Error minimization for sampled noisy \ac{GHZ} states] \thlabel{prob:id2-m}
\begin{subequations}
\begin{align}
 \underset{\Set O, \Set D}{\mathrm{minimize}}\;\;& \max_{\V{\rho}_{\mathrm{all}}\in{\Set S}_{\mathrm{id}}(\V{f}_{\mathrm{all}})}
 \sum_{\SM}\mathbb{E}_{R}\big[(\check{F}-\bar{f}_{\SM})^2\big] 
\label{eqn:minvar_id2}\\
\text{subject to}\;\;
&\mathbb{E}_{R}\big[\check{F}-\bar{f}_{\SM}\big| \V{\rho}_{\SM} \big]=0,\nonumber\\
&\qquad \forall \V{\rho}_{\mathrm{all}}\in{\Set S}_{\mathrm{id}}(\V{f}_{\mathrm{all}}), \; \SM \subset \Set N \label{eqn:unbiased_id2}
\end{align}
\end{subequations}
where $\V{\rho}_{\SM} $ is the density matrix of all sampled noisy \ac{GHZ} states, and ${\Set S}_{\mathrm{id}}(\V{f}_{\mathrm{all}})$ denotes the set of $N$ noisy \ac{GHZ} states with independent noise and fidelity composition $\V{f}_{\mathrm{all}}=\{f_n,n\in \Set N\}$, i.e.,
\begin{align}
\V{\rho}_{\mathrm{all}} = \otimes^{n\in\Set N} \V{\rho}_{n}, 
\end{align}
in which $\V{\rho}_{n}$ is the density matrix of the $n$-th noisy \ac{GHZ} states, 
\begin{align}
f_n= \langle G^{s}_{\V{t}}|\V{\rho}_n | G^{s}_{\V{t}}\rangle, \quad n\in \Set N,
\end{align}
is the fidelity of the $n$-th noisy \ac{GHZ} states, and
\begin{align}
\bar{f}_{\SM}= \frac{1}{M}\sum_{n\in \SM}f_n
\end{align}
is the average fidelity of the sampled noisy \ac{GHZ} states.
\end{Problem}

In $\Ps$-\ref{prob:id-m}, because the noise is independent, the fidelity of the unsampled \ac{GHZ} states is invariant to the measurment outcome.
Consequently, the estimation error of $\Ps$-\ref{prob:id-m} can be expressed as the sum of two parts, namely the estimation error of the sampled \ac{GHZ} states,
\begin{align}
\sum_{\SM}\mathbb{E}_{R}\big[(\check{F}-\bar{f}_{\SM} )^2\big],\label{eqn:error-measure}
\end{align}
and the sampling error, i.e., the deviation between the average fidelity of the sampled and the unsampled \ac{GHZ} states,
\begin{align}
\sum_{\SM}(\bar{f}_{\SM} -\bar{f})^2.\,\label{eqn:error-sample}
\end{align}
The sampling error \eqref{eqn:error-sample} is not affected by the estimation protocol, and hence can be omitted in the estimation error minimization problem.
In this way, $\Ps$-\ref{prob:id-m} is transformed into $\Ps$-\ref{prob:id2-m}.

\thref{lem:id-sample} provides a formal proof of the analysis above.

\begin{Lem}[Equivalence of Problems~\ref{prob:id-m} and \ref{prob:id2-m}] \thlabel{lem:id-sample}
If a measurement operation $\Set O^*$ is optimal in $\Ps$-\ref{prob:id2-m} for all compositions of fidelity, i.e.,
$\forall\, {\Set S}_{\mathrm{id}}(\V{f}_{\mathrm{all}})\subset{\Set S}_{\mathrm{id}}$, 
it is also optimal in $\Ps$-\ref{prob:id-m}.
\end{Lem}
\begin{proof}
Please refer to Appendix~\ref{pf_lem:id-sample}
\end{proof}

The following theorem summarizes the results of this section.
\begin{Thm}[Generality of optimality with independent noise] \thlabel{thm:gen-opt-id}
If a measurement operation $\Set O^*$ is optimal in $\Ps$-\ref{prob:id2-m} for all
${\Set S}_{\mathrm{id}}(\V{f}_{\mathrm{all}})\subset{\Set S}_{\mathrm{id}}$,
then the composite measurement operation $\hat{\Set O}^* = \Set O^* \circ \Set T$ is optimal in $\Ps$-\ref{prob:1-m},
where the operation $\Set T$ is defined as in \thref{def:T}.
\end{Thm} 
\begin{proof}
The theorem is a direct consequence of Lemmas~\ref{lem:sp2arb}, \ref{lem:id2sp}, and \ref{lem:id-sample}.
\end{proof}

\section{Operation Construction}
\label{sec:construction}
In this section, we will first characterize the minimum estimation error of $\Ps$-\ref{prob:id2-m} and then construct a protocol to achieve this minimum error.

\subsection{Lower bound of the measurement error}

The characterization of the minimum mean squared error is a key research focus due to its theoretical and practical significance.
For $\Ps$-\ref{prob:id2-m}, the Cram\'er-Rao bound \cite{Cra:B99,Rao:B94} and the quantum Fisher information \cite{Paris:09,Saf:18}
respectively characterize the minimum estimation error for a given measurement operation $\Set O$,
and for the case where all measurement operations available.

However, $\Ps$-\ref{prob:id2-m} aims to minimize the estimation error with only separable measurement operations.
A new method is needed to characterize the minimum error in this case.
For noisy Bell states, the minimum estimation error was characterized in the author's prior work \cite{Ruan:J23}.
This result is achieved by indentifying the limitations of separable operations when measuring entangled states.
Recognizing that Bell states are a special case of \ac{GHZ} states, we prove in \thref{lem:min} that the minimum estimation error derived in \cite{Ruan:J23} also applies to the estimation of noisy \ac{GHZ} states.

\begin{Lem}[Lower bound for the estimation error]\thlabel{lem:min}
The mean squared estimation error, i.e., the objective function of $\Ps$-\ref{prob:id2-m} divided by $N \choose M$, is no less than 
\begin{align}
\sum_{n\in\Set N}\frac{(2f_n+1)(1-f_n)}{2MN}.
\label{eqn:errorbound}
\end{align}
\end{Lem}

\begin{proof}
Please refer to Appendix~\ref{pf_lem:min} for the proof.
\end{proof}

\subsection{Achieving the minimum estimation error}
This subsection constructs an estimation protocol that serves as the optimal solution of $\Ps$-\ref{prob:1-m}.

A noisy \ac{GHZ} state $\V{\rho}$ with fidelity $f$ can be expressed as
\begin{align}
\V{\rho} = f\M{G}^s_{\V{t}} + (1-f)\M{N},
\end{align}
where the noise component $\M{N}$ satisfies
\begin{align}
\mathrm{Tr}(\M{N}) = 1,\;\; \mathrm{Tr}(\M{G}^s_{\V{t}}\M{N}) = 0,\;\;(\mathbb{I}_{2^L} - \M{G}^s_{\V{t}})\M{N}\succcurlyeq \M{0} .
\label{eqn:N}
\end{align}
To ensure an unbiased estimate, the expectation of the measurement outcome $R$ should be determined by the fidelity $f$ and not influenced by the noise component $\M{N}$.
Furthermore, with no prior information about the noise, $\M{N}$ can be any matrix that satisfies \eqref{eqn:N}.
Therefore, the observable corresponding to the measurement outcome must have the following form:
\begin{align}
\begin{split}
\M{M} = c (\mathbb{I}_{2^L}-\M{G}^s_{\V{t}}),
\end{split}\label{eqn:obM}
\end{align}
where the coefficient $c\in(0,1]$. 
This coefficient represents the probability that the measurement can distinguish the fidelity component $\M{G}^s_{\V{t}}$ from the noise component $\M{N}$. 
In this case, the fidelity of $\V{\rho}$ is given by
\begin{align}
f &= 1- \frac{1}{c}\mathrm{tr}\big[c(\mathbb{I}_{2^L}-\M{G}^s_{\V{t}})\V{\rho}
\big]\nonumber\\
&=1- \frac{1}{c}\Pr[R=1].\label{eqn:fcp}
\end{align}
To minimize the estimation error, $c$ needs to be maximized. 
For the two-node case, i.e., $L=2$, it has been shown that $\frac{2}{3}$ is the maximum  achievable value for $c$ \cite{Ruan:J23}.
Furthermore, based on the analysis in the previous step, the two-node case provides a lower bound on the estimation error for cases with a generic number of nodes.
Therefore, we construct Protocol~\ref{alg:fidelityest} to achieve $c=\frac{2}{3}$.

Specifically, every $L$-qubit state $\V{\rho}$ has a unique Bloch representation:
\begin{align}
\V{\rho} = \frac{\mathbb{I}_{2^L}}{2^L} + \sum_{j=1}^{4^L-1}r^{(j)}\V{\sigma}^{(j)},
\end{align}
where $\V{\sigma}^{(j)}$ are the $L$-qubit non-identity Pauli observables, and the Bloch coefficients $r^{(j)}\in \mathbb{R}$ determine the state. 
Since 
\begin{align}
r^{(j)} = \mathrm{tr}\big[\V{\rho} \V{\sigma}^{(j)}\big],
\end{align} 
the Bloch coefficients can be inferred through local Pauli measurements. 
Therefore, the Bloch representation provides a valuable tool for estimating quantum states via local measurements.

Building on the above insights, \thref{lem:Bloch} gives the Bloch representation of the \ac{GHZ} states.
Then, by applying this representation in conjunction with combinatorial theory, \thref{lem:Paulivalue} characterizes the 
expectation of several weighted Pauli observables when measured on \ac{GHZ} states.
Based on this result, \thref{lem:unbiased} demonstrates  that \eqref{eqn:checkF} is an unbiased estimator for $\Ps$-\ref{prob:id2-m},
and \thref{lem:opt-id} further establishes that this estimator is optimal for $\Ps$-\ref{prob:id2-m}.
Finally, by integrating these findings with \thref{thm:gen-opt-id}, \thref{thm:opt} confirms that Protocol~\ref{alg:fidelityest} is the optimal solution for the original fidelity estimation problem $\Ps$-\ref{prob:1-m}.

\begin{figure}
\vspace{-3mm}
\begin{algorithm}[H]
\caption{Fidelity estimation}\label{alg:fidelityest}
\begin{algorithmic}[1]
%\Statex \vspace{-3mm}%adjust space between caption and operations
\State{\em Set measurement parameters.} 
The nodes select the sampling set $\Set M \subset \Set N$ completely at random, with $|\Set M| =M$.
Additionally, the nodes generate $M$ number of \ac{i.i.d.} random variables $A_n \in\{0,1\}$, $n\in \Set M$,
with distribution 
\begin{align}
\Pr[A_n=0]=\frac{1}{3},\quad \mbox{and} \quad \Pr[A_n=1]=\frac{2}{3}.  \label{eqn:An}
\end{align}
In the case of $A_n = 1$, the nodes randomly pick an $L$-bit string with even $l_1$ norm, i.e.,
\begin{align}
\V{K}_n = [K_{n,1},K_{n,2},\ldots,K_{n,L}]\in \{0,1\}^L, \;\; ||\V{K}_n||_1 \mbox{ even}.\label{eqn:kn}
\end{align}
All strings satisfying \eqref{eqn:kn} are picked with equal probability.
\State{\em Perform measurements.} For the noisy \ac{GHZ} state $n\in \Set{M}$, the $L$ nodes make measurements according to $A_n$ and $\V{K}_n$, and record the measurement result as correct  ($r_n=0$) or erroneous ($r_n = 1$).
Specifically, 
\begin{itemize}[leftmargin=*]
\item[-] When $A_n = 0$, all $L$ nodes measure in the $z$-basis and record the result $\V{c}_n\in\{0,1\}^L$ according to \eqref{eqn:r_n-z}
\begin{align}
r_n = \indicator\big(\V{c}_n \not\in\{\V{t},\tilde{\V{t}}\} \big),\label{eqn:r_n-z}
\end{align}
where $\indicator(\cdot)$ is the indicator function.
\item[-] When $A_n = 1$, the nodes measure in $x$- or $y$-basis to evaluate the Pauli observable 
\begin{align}
\label{eqn:sigmak}
\V{\sigma}^{(\V{K}_n)}_{xy} = \Motimes_{l \in \Set L} \V{\sigma}_{xy}^{(K_{n,l})}, \;\;\mbox{where}\;\; \V{\sigma}_{xy}^{(k)} = \begin{cases}
\V{\sigma}_y, &\hspace{-1.5mm} \mbox{if } k = 1 \\
\V{\sigma}_x,& \hspace{-1.5mm} \mbox{if } k = 0 \\
\end{cases}\hspace{-1.5mm}
\end{align}
and get result $c_n\in\{-1,1\}$.
The nodes record the result according to \eqref{eqn:r_n-xy}.
\begin{align}
r_n = \left\{\rule{0cm}{0.6cm}\right.
\hspace{-.5mm}\begin{array}{l@{\,}l@{\;}l}
\hspace{-1mm}\indicator\big(c_n =&(-1)^{\frac12||\V{K}_n||_1 + \V{K}_n \cdot \V{t}-1}\big)&\mbox{if }  s  = + \vspace{1mm}\\
\hspace{-1mm}\indicator\big(c_n =&(-1)^{\frac12||\V{K}_n||_1 + \V{K}_n \cdot \V{t}}\big)&\mbox{if }  s  = -
\end{array},\hspace{-3mm}
\label{eqn:r_n-xy}
\end{align}
where $\V{K}_n \cdot \V{t}$ is the inner product of two bit strings.
\end{itemize}
\State{\em Estimate fidelity.} The number of errors and the \ac{QBER} are expressed as
$
e_{\SM} = \sum_{n\in\SM} r_{n}$, and
$\varepsilon_{\SM}= \dfrac{e_{\SM}}{M}$,
respectively. The estimated fidelity is given by
\begin{align}
\check{f}= 1- \frac{3}{2}\varepsilon_{\SM}.\label{eqn:checkF}
\end{align}
\end{algorithmic}
\end{algorithm}
\end{figure}

%\newpage
\begin{Remark}[Implementability of Protocol~\ref{alg:fidelityest}]
\thlabel{rem:implement}
Protocol~\ref{alg:fidelityest} exclusively relies on local Pauli measurements, which are standard operations implementable across a variety of quantum platforms \cite{DeGetal:J13,Yinetal:J17,Kaletal:J17,Bocetal:J18}.~\hfill~\QEDB
\end{Remark}

\begin{Lem}[Bloch representation of \ac{GHZ} states]
\thlabel{lem:Bloch}
The Bloch representation of a \ac{GHZ} state is given by
\begin{align}
\M{G}^{\pm}_{\V{j}}
&= \frac{1}{2^L} \sum_{\substack{ \V{k} \in \{0,1\}^L \\ ||\V{k}||_1\scriptsize \mbox{ even}}} (-1)^{\V{k}\cdot\V{j}}
\Big( \V{\sigma}_{Iz}^{(\V{k})} \pm (-1)^{\frac{||\V{k}||_1}{2}}\V{\sigma}_{xy}^{(\V{k})}\Big),\label{eqn:gGHZBloch}
\end{align}
where 
\begin{align}
\V{\sigma}_{Iz}^{(\V{k})} = \Motimes_{l\in\Set L} \V{\sigma}_{Iz}^{(k_l)},\quad \V{\sigma}_{xy}^{(\V{k})} = \Motimes_{l\in\Set L} \V{\sigma}_{xy}^{(k_l)},   
\end{align}
$k_l$ is the $l$-th element of $\V{k}$, and%$|\V{k}| = \sum_{l \in \Set L} \V{k}[l]$ is the Hamming weight of a bit string and
\begin{equation}
\V{\sigma}_{Iz}^{(k_l)} = \begin{cases}
\V{\sigma}_z, & \mbox{if } k_l = 1, \\
\mathbb{I}_2, & \mbox{if } k_l = 0, \\
\end{cases},
\quad
\V{\sigma}_{xy}^{(k_l)} = \begin{cases}
\V{\sigma}_y, & \mbox{if } k_l = 1, \\
\V{\sigma}_x, & \mbox{if } k_l = 0, \\
\end{cases}\label{eqn:sigmaxyz}
\end{equation}
in which $\V{\sigma}_x$, $\V{\sigma}_y$, and $\V{\sigma}_z$ are the Pauli matrices.
\end{Lem}
\begin{proof}
Please refer to Appendix~\ref{pf_lem:Bloch}
\end{proof}

%Based on \thref{lem:Bloch}, \thref{lem:Paulivalue} characterizes the value of a set of weighted Pauli observables when measured on the \ac{GHZ} states.

\begin{Lem}[Value of weighted Pauli observables]
\thlabel{lem:Paulivalue}
Consider the Pauli observables
\begin{equation}
\label{eqn:sigmak-2}
\V{\sigma}^{(\V{k})}_{xy}  =  \Motimes_{l \in \Set L} \V{\sigma}_{xy}^{(k_l)}, \qquad \V{\sigma}_{xy}^{(k_l)}  =  \begin{cases}
\V{\sigma}_y, & \mbox{if } k_l = 1, \\
\V{\sigma}_x, & \mbox{if } k_l = 0, \\
\end{cases}
\end{equation}
where $\V{k}=[k_1,k_2,\ldots,k_L]\in\{0,1\}^L$ is an $L$-bit string with even $l_1$ norm.
When these observables are measured on the \ac{GHZ} states, the expected values have the following properties.
\begin{align}
 \mathrm{Tr}\bigg[
 \Big(\frac{1}{2^{L-1}} \sum_{ \V{k} \in \{0,1\}^L \atop ||\V{k}||_1\scriptsize \mbox{ even}}  (-1)^{\frac{||\V{k}||_1}{2} + \V{k}\cdot\V{j}} 
 \V{\sigma}_{xy}^{(\V{k})}\Big)
  \M{G}^{\pm}_{\V{j}}\bigg] =\pm1,\label{eqn:obv_sumspm_w}\\
   \forall \V{j} \in {\Set J},\mbox{ and}\nonumber\\
\mathrm{Tr}\bigg[
\Big(\frac{1}{2^{L-1}} \sum_{\scriptsize \V{k} \in \{0,1\}^L \atop ||\V{k}||_1 \mbox{ even}}  (-1)^{\frac{||\V{k}||_1}{2} + \V{k}\cdot\V{j}} 
 \V{\sigma}_{xy}^{(\V{k})}\Big)
\M{G}^{\pm}_{\TVj}\bigg] =0,\hspace{2.8mm}\label{eqn:obv_sums0_w}\\
\forall \TVj \in {\Set J},\; \TVj \neq \V{j}.\nonumber
\end{align}
\end{Lem}
\begin{proof}
Please refer to Appendix~\ref{pf_lem:Paulivalue}
\end{proof}

\begin{Lem}[Unbiased estimator]
\thlabel{lem:unbiased}
If Protocol~\ref{alg:fidelityest} is applied to slove $\Ps$-\ref{prob:id2-m}, the estimated fidelity $\check{f}$ given by \eqref{eqn:checkF} satisfies \eqref{eqn:unbiased_id2}.
\end{Lem}
\begin{proof}
Please refer to Appendix~\ref{pf_lem:unbiased}.
\end{proof}

\begin{Lem}[Optimality with independent noise]\thlabel{lem:opt-id} Protocol~\ref{alg:fidelityest} achieves the minimum estimation error in $\Ps$-\ref{prob:id2-m}.
\end{Lem}
\begin{proof}
Please refer to Appendix~\ref{pf_lem:opt-id}.
\end{proof}

\begin{Thm}[Optimal estimation protocol]\thlabel{thm:opt}
Protocol~\ref{alg:fidelityest} is optimal for $\Ps$-\ref{prob:1-m}.
\end{Thm}
\begin{proof}
Please refer to Appendix~\ref{pf_thm:opt}.
\end{proof}

\section{Application example}

\label{sec:sim}

This section evaluates the mean squared error of the proposed protocol and compares it with existing fidelity estimation methods.
The noise model is described first.

\noindent {\bf Noise model:}
The simulated platform consists of three physically separated \ac{NV} centers.
Within these \ac{NV} centers, the electronic spin associated with the defect can be utilized as a qubit. 
The three \ac{NV} centers aim to generate tripartite \ac{GHZ} states.

Multiple types of noise are simulated in accordance to the experiment results in \cite{AwsHanWraZho:J18,Humetal:J18}.
In particular, the evolution of dark counts is modeled as a Markov chain.
The dark count state is defined as $D_n\in{0,1}$, 
where $D_n = 0$ indicates no dark count in the $n$-th round of \ac{GHZ} state generation, 
and $D_n = 1$ indicates the presence of a dark count.
The transition matrix of the dark count states is defined as follows:
\begin{align}
\M{T} = \begin{bmatrix}1-\delta P_{\mathrm d}&\delta (1-P_{\mathrm d})\\ \delta P_{\mathrm d} &1-\delta (1-P_{\mathrm d})\end{bmatrix}
\label{eqn:MarkovT}
\end{align}
where $P_{\mathrm d}\in[0,1]$ represents the average probability of a dark count occurring in a single round of \ac{GHZ} state generation.
This parameter influences the noise intensity.
The parameter
\begin{align}\delta\in\Big(0,\min(\frac{1}{P_{\mathrm d}},\frac{1}{1-P_{\mathrm d}})\Big]
\end{align}
determines the correlation between two consecutive dark count states, i.e.,
\begin{align}
\mathrm{cor}(D_n,D_{n+1})= 1-\delta.
\end{align}
In particular, when $\delta = 1$, the dark count states are independent random variables.~\hfill~\QEDB

\begin{figure} \centering
% y ticks
\psfrag{0.0}[Br][Br][0.7]{0\hspace{0.2mm}}
\psfrag{1}[Br][Br][0.7]{1\hspace{0.2mm}}
\psfrag{2}[Br][Br][0.7]{2\hspace{0.2mm}}
\psfrag{3}[Br][Br][0.7]{3\hspace{0.2mm}}
\psfrag{4}[Br][Br][0.7]{4\hspace{0.2mm}}
\psfrag{5}[Br][Br][0.7]{5\hspace{0.2mm}}
\psfrag{6}[Br][Br][0.7]{6\hspace{0.2mm}}
\psfrag{7}[Br][Br][0.7]{7\hspace{0.2mm}}
\psfrag{8}[Br][Br][0.7]{8\hspace{0.2mm}}
\psfrag{10}[Br][Br][0.7]{10\hspace{0.2mm}}
\psfrag{12}[Br][Br][0.7]{12\hspace{0.2mm}}
\psfrag{14}[Br][Br][0.7]{14\hspace{0.2mm}}
\psfrag{16}[Br][Br][0.7]{16\hspace{0.2mm}}
\psfrag{10-4}[bl][bl][0.7]{$\times 10^{-4}$}
% x ticks
\psfrag{0.5}[tt][tt][0.7]{0.5}
\psfrag{0.4}[tt][tt][0.7]{0.4}
\psfrag{0.3}[tt][tt][0.7]{0.3}
\psfrag{0.2}[tt][tt][0.7]{0.2}
\psfrag{0.1}[tt][tt][0.7]{0.1}
\psfrag{0}[tt][tt][0.7]{0}
\psfrag{-0.5}[tt][tt][0.7]{-0.5}
%x and y label
\psfrag{A}[bc][bc][0.8]{A}
\psfrag{B}[bc][bc][0.8]{B}
\psfrag{p}[tc][tc][0.8]{$P_{\mathrm{d}}$}
\psfrag{1-delta}[tc][tc][0.8]{$1-\delta$}
\psfrag{varf}[tc][tc][0.8]{$\mathrm{Var}(\check{F})$}
%legend
\psfrag{proposed              1}[cl][cl][0.73]{\hspace{-1mm}\parbox[l]{16mm}{\linespread{0.6}\selectfont  proposed}}
\psfrag{Guh                                 1}[cl][cl][0.73]{\hspace{-1mm}\parbox[l]{23mm}{\linespread{0.6}\selectfont  protocol in  \cite{GuhLuGaoPan:J07}}}
\psfrag{Fla                             1}[cl][cl][0.73]{\hspace{-1mm}\parbox[l]{23mm}{\linespread{0.6}\selectfont  protocol in  \cite{FlaLiu:J11}}}
\hspace*{-6.2mm}\includegraphics[scale=0.5]{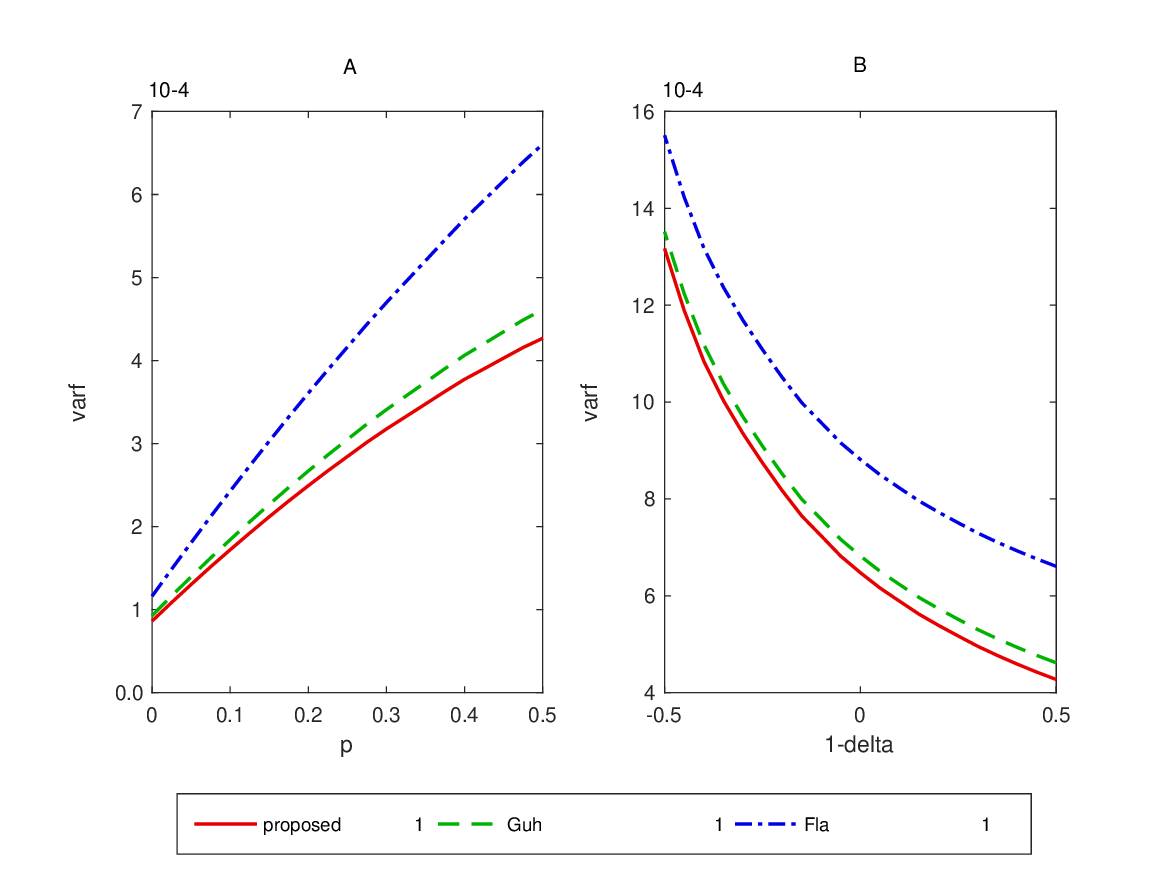}
\caption {\small The variance of the estimated fidelity provided by different protocols. 
In this figure, $N=2000$ and $M=1000$. In subfigure A, $\delta=0.5$, while in subfigure B, $P_{\mathrm{d}}=0.5$.
}
\label{fig_Measure_Var}
\end{figure}

In Fig.~\ref{fig_Measure_Var}, we plot the mean squared error $\mathrm{Var}(\check{F})$ of the proposed protocol, 
along with those of two existing fidelity estimation protocols, i.e., protocols described in  \cite{FlaLiu:J11,GuhLuGaoPan:J07}, 
as functions of the dark count parameters $P_{\mathrm d}$ and $\delta$.

Fig.~\ref{fig_Measure_Var}A illustrates that as the noise intensity $P_{\mathrm d}$ increases, the mean squared error of all protocols increases.
This occurs because stronger noise introduces greater uncertainty in the measurement results.
Consequently, the estimation error of the sampled \ac{GHZ} states, as defined in \eqref{eqn:error-measure}, increases.

Fig.~\ref{fig_Measure_Var}B illustrates that as the noise correlation $1-\delta$ between two consecutive rounds shifts from negative to positive, the mean squared error of all protocols decreases.
This occurs because positive correlation results in a smaller difference between the fidelity of sampled and unsampled \ac{GHZ} states.
As a result, the sampling error, as defined in \eqref{eqn:error-sample}, decreases.

The protocol in \cite{FlaLiu:J11} utilizes only local Pauli measurements, thereby ensuring good implementability.
However, as shown in Fig.~\ref{fig_Measure_Var}, it exhibits a relatively high mean squared error.
The protocol in \cite{GuhLuGaoPan:J07} achieves a lower mean squared error.
However, this protocol necessitates local measurements of the form $\cos(\theta)\V{\sigma}_x + \sin(\theta)\V{\sigma}_y$,
where 
\begin{align*}\theta\in\Big\{\frac{\pi}{L},\frac{2\pi}{L},\ldots,\frac{\pi(L-1)}{L}\Big\}.
\end{align*}
Notably, since the value of $\theta$ must be adjusted by $\frac{\pi}{L}$ for each measurement operator, 
these operators become increasingly difficult to implement as the number of qubits $L$ increases.
Similar to the protocol in \cite{FlaLiu:J11}, the proposed protocol employs only local Pauli measurements,
and Fig.~\ref{fig_Measure_Var} demonstrates that it yields the minimum mean squared error across all cases.
Thus, the proposed protocol simultaneously achieves both good implementability and high estimation accuracy.

\section{Conclusion}
\label{sec:summary} 
Fidelity estimation of entangled states shared by remote nodes is a critical component of quantum networks.
In this work, we propose a fidelity estimation protocol for \ac{GHZ} states under arbitrary noise conditions.
By utilizing the Bloch representation of \ac{GHZ} states and applying tools from combinatorial theory, we design a set of local measurement operators to estimate the fidelity of \ac{GHZ} states.
These operators achieve optimal measurement efficiency within the constraint of separable operations.
Numerical tests confirm the advantages of the proposed estimation protocol.

\section*{Acknowledgement}
 
The author would like to thank Stephanie Wehner, Bas Dirkse, and Sophie Hermans at TUDelft, The Netherlands, for refining the structure of the paper and stimulating discussions on the treatment of non-\ac{i.i.d.} noise.

This work is supported, in part, by the EU Flagship on Quantum Technologies and the Quantum Internet Alliance (funded by the EU Horizon 2020 Research and Innovation Programme, Grant No. 820445).

\appendix

\section{Proof of \thref{lem:sp2arb}}
\label{pf_lem:sp2arb}
We will first demonstrate that the operation $\Set T$ transforms any $L$-qubit state into a state that is diagonal in the basis of \ac{GHZ} states 
while preserving the original state's fidelity.
Then by leveraging this property of $\Set T$, we construct a Gedankenexperiment to establish the equivalence of $\Ps$-\ref{prob:1-m} and $\Ps$-\ref{prob:sp-m}.

The \ac{GHZ} states in set 
\begin{align}
{\Set G}_L=\{|\M{G}^\tau_{\V{j}}\rangle, \V{j}\in \Set J, \tau \{+,-\} \}
\end{align}
form a basis of $L$-qubit states. 
Therefore, the density matrix of any $L$-qubit state can be written as
\begin{align}
\V{\rho} = \sum_{k\in \Set K} c_{k} |\phi_{k}\rangle\langle \psi_{k} |,
\label{eqn:Rhon}
\end{align}
where $|\phi_{k}\rangle, |\psi_{k}\rangle\in {\Set G}_L$ and $\Set K$ is the set of all terms with non-zero coefficient $c_{k}\in\mathbb{C}$.
Given \eqref{eqn:Rhon}, the fidelity of $\V{\rho}$ is given by
\begin{align}
f = \sum_{k\in \Set K}&\Big( c_{k}
\mathds{1}\big(|\phi_{k}\rangle= |\psi_{k}\rangle= |G^{s}_{\V{t}}\rangle \big) \Big), 
\label{eqn:fid}
\end{align}
where $\mathds{1}(\cdot)$ is the indicator function.

Because $\V{\sigma}_x|0\rangle =|1\rangle$, $\V{\sigma}_x|1\rangle =|0\rangle$, $\V{\sigma}_y|0\rangle =\imath |1\rangle$, and $\V{\sigma}_y|1\rangle =-\imath |0\rangle$, 
for any state $|G^{\pm}_{\V{j}}\rangle \in {\Set G}_L$ and any operation $\M{T}^{(\V{k})}\in \Set{T}_{\mathrm{all}}$, we have that
\begin{subequations}
\begin{align}
\M{T}^{(\V{k})} |G^{\pm}_{\V{j}}\rangle 
&= \M{T}^{(\V{k})} \frac{1}{\sqrt{2}}\big(|\V{j}\rangle \pm  |\TVj\rangle\big)\\
&= \frac{1}{\sqrt{2}}\imath^{|\V{k}|} \big((-1)^{\V{k}\cdot\V{j}} |\TVj\rangle \pm  (-1)^{\V{k}\cdot \TVj} |\V{j}\rangle\big)\label{TkG-b}\\
&= \frac{1}{\sqrt{2}}\imath^{|\V{k}|} (-1)^{\V{k}\cdot\V{j}}\big( |\TVj\rangle \pm  |\V{j}\rangle\big)\label{TkG-c},
\end{align}\label{TkG}
\end{subequations}
where $\V{a} \cdot \V{b}$ denotes the inner product of strings $\V{a}$ and $\V{b}$.
Since
$
\V{k}\cdot\V{j} + \V{k}\cdot \TVj = |\V{k}|%\label{eqn:kj+kbarj}
$
and $|\V{k}|$ is even, it follows that
\begin{align}
(-1)^{\V{k}\cdot\V{j}} =  (-1)^{\V{k}\cdot \TVj}.\label{eqn:kj=kbarj}
\end{align}
Therefore, \eqref{TkG-c} holds.

From \eqref{TkG}, it follows that for any two states $|G^{ \tau_a}_{\V{j}_a}\rangle, |G^{\tau_b}_{\V{j}_b}\rangle \in  {\Set G}_L$ and any operation $\M{T}^{(\V{k})}\in \Set{T}_{\mathrm{all}}$,
\begin{align}
&\M{T}^{(\V{k})} |G^{\tau_a}_{\V{j}_a}\rangle\langle G^{\tau_b}_{\V{j}_b}| \M{T}^{(\V{k})\dag}\nonumber\\
&=\left\{
\begin{array}{ll} 
 {\color{white}-}(-1)^{\V{k}\cdot \V{j}_a + \V{k}\cdot \V{j}_b}|G^{\tau_a}_{\V{j}_a}\rangle\langle G^{\tau_b}_{\V{j}_b}| & \mbox{ if } \tau_a =\tau_b \vspace{1.5mm}\\
 -(-1)^{\V{k}\cdot \V{j}_a + \V{k}\cdot \V{j}_b}|G^{\tau_a}_{\V{j}_a}\rangle\langle G^{\tau_b}_{\V{j}_b}| & \mbox{ if } \tau_a \neq \tau_b 
 \end{array}
 \right..\label{eqn:TkGj1j2}
\end{align}

From \eqref{eqn:Tk} and \eqref{eqn:TkGj1j2}, it is easy to see that 
\begin{align}
&\Set{T}^{(\V{k})} \big(|G^{\tau_a}_{\V{j}_a}\rangle\langle G^{\tau_b}_{\V{j}_b}|\big)\nonumber\\
&=\frac{1}{2}\big(\mathbb{I}_{2^L} |G^{\tau_a}_{\V{j}_a}\rangle\langle G^{\tau_b}_{\V{j}_b}| \mathbb{I}_{2^L}^\dag +\M{T}^{(\V{k})} |G^{\tau_a}_{\V{j}_a}\rangle\langle G^{\tau_b}_{\V{j}_b}| \M{T}^{(\V{k})\dag}\big)\nonumber\\
&=\left\{
\begin{array}{ll@{\;}l}|G^{\tau_a}_{\V{j}_a}\rangle\langle G^{\tau_b}_{\V{j}_b}| &\mbox{if }\hspace{1.1mm} (-1)^{\V{k}\cdot \V{j}_a + \V{k}\cdot \V{j}_b} =1 &\&\; \tau_a = \tau_b\\
&\mbox{or } (-1)^{\V{k}\cdot \V{j}_a + \V{k}\cdot \V{j}_b} =-1 &\&\; \tau_a \neq \tau_b\vspace{1mm}\\
\M{0} & \mbox{otherwise} \end{array}\right..\label{eqn:T-all}
\end{align}

From \eqref{eqn:T-all}, it is evident that for all diagonal terms in the basis of \ac{GHZ} states, because $\V{j}_a= \V{j}_b$ and $\tau_a = \tau_b$,
\begin{align}
\Set{T}^{(\V{k})} \big(\M{G}^{\pm}_{\V{j}}\big)
&=\M{G}^{\pm}_{\V{j}}, \quad \forall\, |{G}^{\pm}_{\V{j}}\rangle\in {\Set G}_L,\,  \M{T}^{(\V{k})}\in \Set{T}_{\mathrm{all}}.\label{eqn:T-diag}
\end{align}
From \eqref{eqn:fid} and \eqref{eqn:T-diag} , operation $\Set{T}$ does not change the fidelity of any $L$-qubit state $\V{\rho}$.

Recall the definition of the binary strings $\V{j}$
\begin{align}
\V{j} \in {\Set J} = \{0\}\! \frown\! \{0,1\}^{L-1}\label{eqn:Vj-a},
\end{align}
where $\frown $ denotes string concatenation.
For two different strings $\V{j}_a,\V{j}_b \in {\Set J}$, denote 
\begin{align}
\V{j}_{a,b} =  \{1\}\! \frown \!\{0\}^{l-2}  \! \frown \! \{1\}\! \frown \! \{0\}^{L-l},\label{eqn:Jab}
\end{align}
where $l$ is a bit at which the elements of $\V{j}_a$ and $\V{j}_b$ mismatch, i.e.,
\begin{align}
j_{a,l} \oplus j_{b,l} = 1,
\end{align}
in which $\oplus$ denotes the XOR operation.

For two distinct states $|G^{\tau_a}_{\V{j}_a}\rangle, |G^{\tau_b}_{\V{j}_b}\rangle \in  {\Set G}_L$,
either $\V{j}_a \neq \V{j}_b$ or $\tau_a \neq \tau_b$.
In this case, the following binary string can be defined 
\begin{align}
\V{k}_{a,b} = \left\{\begin{array}{ll}
\V{j}_{a,b}  & \mbox{if } \tau_a = \tau_b\\
\{0\}^L & \mbox{if } \tau_a \neq \tau_b
\end{array}
\right..\label{eqn:kab}
\end{align}
According to \eqref{eqn:T_all}, $\Set{T}^{(\V{k}_{a,b})}\in {\Set T}_{\mathrm{all}}$.

From \eqref{eqn:kab}, we have
\begin{align}
\V{k}_{a,b}\cdot \V{j}_a + \V{k}_{a,b}\cdot \V{j}_b =
\left\{\begin{array}{ll}
1 &\mbox{if } \tau_a = \tau_b\\
0&\mbox{if } \tau_a \neq \tau_b
\end{array}\right..\label{eqn:kj}
\end{align}
By substituting \eqref{eqn:kj} into \eqref{eqn:T-all}, it can be obtained that
\begin{align}
\Set{T}^{(\V{k}_{a,b})} \big(|G^{\tau_a}_{\V{j}_a}\rangle\langle G^{\tau_b}_{\V{j}_b}|\big) = \M{0}, \; \forall\; |G^{\tau_a}_{\V{j}_a}\rangle \neq |G^{\tau_b}_{\V{j}_b}\rangle \in {\Set G}_L.\label{eqn:T-offdiag}
\end{align}

The operation $\Set T$ sequentially applies each  $\Set{T}^{(\V{k})}\in {\Set T}_{\mathrm{all}}$.
Therefore, \eqref{eqn:T-diag} and \eqref{eqn:T-offdiag} indicate that
\begin{align}
\Set{T} \big(|G^{\tau_a}_{\V{j}_a}\rangle\langle G^{\tau_b}_{\V{j}_b}|\big)=\left\{
\begin{array}{ll}|G^{\tau_a}_{\V{j}_a}\rangle\langle G^{\tau_b}_{\V{j}_b}| &\mbox{if }|G^{\tau_a}_{\V{j}_a}\rangle = |G^{\tau_b}_{\V{j}_b}\rangle
\vspace{1mm}\\
\M{0} & \mbox{otherwise} \end{array}\right..\label{eqn:Tall-all}
\end{align}

Next, consider a Gedankenexperiment  where the operation $\Set T$ is performed on the sampled states $n\in \SM$ before the measurement $\Set O$.
Subsequently, this rotation is applied to the unsampled states $n \in \Set N \backslash \SM$. 
According to \eqref{eqn:Tall-all}, the rotation $\Set T$ on the unsampled states neither affects the measurement outcome $\V{r}$ nor the average fidelity of the unsampled states $\bar{f}$.
In this context, the operation considered in the Gedankenexperiment is equivalent to the operation $\hat{\Set O} = \Set O \circ \Set T$, i.e., both the rotation and the measurement are performed only on the sampled states.

Furthermore, since the rotation $\Set T$ on the unsampled states commutes with the measurement  $\Set O$, the operation $\hat{\Set O}$ is also equivalent to an operation
where the rotation $\Set T$ is applied to all noisy \ac{GHZ} states prior to measurement $\Set O$.
Given these two equivalences, to analyze the performance of the operation $\hat{\Set O}$, it is equivalent to consider the scenario in which the rotation $\Set T$ is applied to all noisy \ac{GHZ} states before measurement.

An arbitrary $N\times L$ qubits state can be expressed in the \ac{GHZ} states basis as follows.
\begin{align}
\V{\rho}_{\mathrm{all}} = \sum_{k\in \Set K} c_k\otimes^{n\in\Set N} |\phi_{k,n}\rangle\langle \psi_{k,n} |,
\label{eqn:Rhoall}
\end{align}
where $|\phi_{k,n}\rangle, |\psi_{k,n}\rangle$ are \ac{GHZ} states, i.e., $|\phi_{k,n}\rangle, |\psi_{k,n}\rangle \in {\Set G}_L$, and $\Set K$ is the set of all terms with non-zero coefficients $c_k\in\mathbb{C}$.
Substituting \eqref{eqn:Tall-all} into \eqref{eqn:Rhoall}, we find that after applying the rotation $\Set T$ to the qubits, their joint state is
\begin{align}
\Set T(\V{\rho}_{\mathrm{all}}) = \sum_{k\in \Set K} c_k\otimes^{n\in\Set N}\mathds{1}\big(|\phi_{k,n}\rangle= |\psi_{k,n}\rangle \big)|\phi_{k,n}\rangle\langle \psi_{k,n} |.
\label{eqn:Rhoall-R}
\end{align}
Equation \eqref{eqn:Rhoall-R} demonstrates that after applying $\Set T$ on all $N\times L$ qubits, their state itheir state becomes a composition of product states, making it separable across different $L$-qubit groups, i.e.,
\begin{align}
\Set T(\V{\rho}_{\mathrm{all}}) \in \Set{S}_{\mathrm{sp}},\quad \forall \V{\rho}_{\mathrm{all}}. \label{eqn:Trhosp}
\end{align}

The function $e(\Set S, \Set O, \Set D)$ is defined as the worst-case squared estimation error for states in the set $\Set S$, given the measurement operator $\Set O$ and the estimator $\Set D$, i.e.
\begin{align}
e(\Set S, \Set O, \Set D)
= \max_{\V{\rho}_{\mathrm{all}}\in{\Set S}}\sum_{\SM}\mathbb{E}_{R}\big[(\check{F}-\bar{F})^2\big|\SM, \Set O, \Set D\big].
\label{eqn:fSOD}
\end{align}
Then the objective functions of $\Ps$-\ref{prob:1-m} and $\Ps$-\ref{prob:sp-m} can be rewritten as
\begin{align*}
\underset{\Set O, \Set D}{\mathrm{minimize}}\, e({\Set S}_{\mathrm{arb}}, \Set O, \Set D), \mbox{ and } 
\underset{\Set O, \Set D}{\mathrm{minimize}}\, e(\Set S_{\mathrm{sp}}, \Set O, \Set D),
\end{align*}
respectively.
 
Let $\Set D^*$ denote the estimator that, in conjunction with measurement operation $\Set O^*$, achieves the minimum mean squared error.
Then the following inequalities involving the objective functions of $\Ps$-\ref{prob:1-m} and $\Ps$-\ref{prob:sp-m} can be derived:
\begin{subequations}
\begin{align}
\underset{\Set O, \Set D}{\mathrm{minimize}}\, &e({\Set S}_{\mathrm{arb}}, \Set O, \Set D)\nonumber\\
\leq &e({\Set S}_{\mathrm{arb}}, \hat{\Set O}^*, \Set D^*) \label{eqn:arb-le-sp-a}\\
 \leq &e(\Set S_{\mathrm{sp}}, {\Set O}^*, \Set D^*) \label{eqn:arb-le-sp-b}\\
= & \underset{\Set O, \Set D}{\mathrm{minimize}}\,e(\Set S_{\mathrm{sp}}, \Set O, \Set D), 
\label{eqn:arb-le-sp-c}
\end{align} \label{eqn:arb-le-sp}
\end{subequations}
where \eqref{eqn:arb-le-sp-b} holds because of \eqref{eqn:Trhosp}, i.e., $\Set T$ converts an arbitrary state to a separable state,
 \eqref{eqn:arb-le-sp-c} holds because ${\Set O}^*, \Set D^*$ is the optimal solution to $\Ps$-\ref{prob:sp-m}.

Moreover, because ${\Set S}_{\mathrm{arb}} \supset \Set S_{\mathrm{sp}}$,
\begin{align}
\underset{\Set O, \Set D}{\mathrm{minimize}}\,e({\Set S}_{\mathrm{arb}}, \Set O, \Set D) 
\geq \underset{\Set O, \Set D}{\mathrm{minimize}}\, e(\Set S_{\mathrm{sp}}, \Set O, \Set D).
\label{eqn:arb-ge-sp}
\end{align}
According to \eqref{eqn:arb-le-sp} and \eqref{eqn:arb-ge-sp}, one can get that
\begin{align}
\begin{split}
\underset{\Set O, \Set D}{\mathrm{minimize}}\, e({\Set S}_{\mathrm{arb}}, \Set O, \Set D)
&= e({\Set S}_{\mathrm{arb}}, \hat{\Set O}^*, \Set D^*)
\\
&= \underset{\Set O, \Set D}{\mathrm{minimize}}\, e(\Set S_{\mathrm{sp}}, \Set O, \Set D).
\end{split}
\label{eqn:arb-eq-sp}
\end{align}
\eqref{eqn:arb-eq-sp} shows that $\hat{\Set O}^*$ is optimal for $\Ps$-\ref{prob:1-m}.
This completes the proof of \thref{lem:sp2arb}.

\section{Proof of \thref{lem:id2sp}}
\label{pf_lem:id2sp}
Denote $\check{f}(r)$ as the estimated fidelity given the measurement outcome $r$, 
and denote $\bar{f}(\V{\rho}_{\mathrm{all}},\Set{M}, r)$ as the average fidelity of the 
unsampled states given the state $\V{\rho}_{\mathrm{all}}$, sample set $\SM$ and measurement outcome $r$.
According to \eqref{eqn:Rall-sep}, the following inequalities hold for any estimation protocol $\{\Set O$, $\Set D\}$ and any state $\V{\rho}_{\mathrm{all}}\in \Set{S}_{\mathrm{sp}}$. 
\begin{subequations} 
\begin{align}
&\sum_{\SM}\mathbb{E}_{R}\big[(\check{F}-\bar{F})^2\big]\nonumber \\
&=  \sum_{\SM }\sum_{r}\mathrm{Pr}(r|\SM)
\big(\check{f}(r) -\bar{f}(\V{\rho}_{\mathrm{all}},\Set{M}, r)\big)^2\label{eqn:sp2id-a}\\
&= \sum_{\SM }\sum_{r }\mathrm{Pr}(r|\SM)\big(\check{f}(r) -\nonumber\\
&\hspace{13mm} \sum_{k }
\mathrm{Pr}(\V{\rho}^{(k)}_{\mathrm{all}}| \SM, r) \bar{f}(\V{\rho}^{(k)}_{\mathrm{all}},\Set{M}, r) \big)^2 \label{eqn:sp2id-b}\\
&\leq  \sum_{\SM }\sum_{r }\sum_{k}
\mathrm{Pr}(r|\SM)\mathrm{Pr}(\V{\rho}^{(k)}_{\mathrm{all}}| \SM, r)\nonumber\\
&\hspace{28.5mm} 
\big(\check{f}(r) -\bar{f}(\V{\rho}^{(k)}_{\mathrm{all}},\Set{M}, r) \big)^2 \label{eqn:sp2id-c}\\
&= \sum_{k }p_k\sum_{\SM }\sum_{r }
\mathrm{Pr}( r|\V{\rho}^{(k)}_{\mathrm{all}},\SM)\nonumber\\
&\hspace{28.5mm} 
\big(\check{f}(r) - \bar{f}(\V{\rho}^{(k)}_{\mathrm{all}},\Set{M}, r)\big)^2 \label{eqn:sp2id-d}\\
&= \sum_{k }p_k\sum_{\SM}\mathbb{E}_{R}\big[(\check{F}-\bar{f})^2\big|\V{\rho}^{(k)}_{\mathrm{all}}\big], \label{eqn:sp2id-e}\\
&\leq  \max_{\V{\rho}_{\mathrm{all}}\in{\Set S}_{\mathrm{id}}}\sum_{\SM}\mathbb{E}_{R}\big[(\check{F}-\bar{f})^2\big], \label{eqn:sp2id-f}
\end{align}\label{eqn:sp2id}
\end{subequations}
where \eqref{eqn:sp2id-c} holds because $x^2$ is a convex function and
\begin{align}
\begin{split}
\mathrm{Pr}(\V{\rho}^{(k)}_{\mathrm{all}}| \SM, r)&\geq 0, \quad \forall k \in \Set K,\\
\sum_{k \in \Set K}\mathrm{Pr}(\V{\rho}^{(k)}_{\mathrm{all}}| \SM, r) &=1,
\end{split}
\end{align} 
\eqref{eqn:sp2id-d} holds because according to Bayes' theorem
\begin{align}
\mathrm{Pr}(r|\SM)\mathrm{Pr}(\V{\rho}^{(k)}_{\mathrm{all}}| \SM, r) &=\mathrm{Pr}(\V{\rho}^{(k)}_{\mathrm{all}}|\SM)\mathrm{Pr}(r|\V{\rho}^{(k)}_{\mathrm{all}},\SM)
\nonumber\\
&=p_k \mathrm{Pr}(\SM, r|\V{\rho}^{(k)}_{\mathrm{all}}),
\end{align}
and \eqref{eqn:sp2id-f} holds because $\V{\rho}^{(k)}_{\mathrm{all}} \in \Set{S}_{\mathrm{id}}$, $p_k\geq 0$,  $\forall k \in \Set K$, and $\sum_{k}p_k=1$.

Recall the worst-case squared estimation error function $e(\Set S, \Set O, \Set D)$ defined in \eqref{eqn:fSOD}, and
denote $\Set D^*$ as the estimator that, together with the measurement operation $\Set O^*$, minimizes the mean squared error.
Then, the following inequalities involving the objective functions of $\Ps$-\ref{prob:sp-m} and $\Ps$-\ref{prob:id-m} can be derived:
\begin{subequations}
\begin{align}
\underset{\Set O, \Set D}{\mathrm{minimize}}\, &e(\Set S_{\mathrm{sp}}, \Set O, \Set D)\nonumber\\
 \leq &e(\Set S_{\mathrm{sp}}, {\Set O}^*, \Set D^*) \label{eqn:sp-le-id-a}\\
  \leq &e(\Set S_{\mathrm{id}}, {\Set O}^*, \Set D^*) \label{eqn:sp-le-id-b}\\
= & \underset{\Set O, \Set D}{\mathrm{minimize}}\,e(\Set S_{\mathrm{id}}, \Set O, \Set D), 
\label{eqn:sp-le-id-c}
\end{align} \label{eqn:sp-le-id}
\end{subequations}
where \eqref{eqn:sp-le-id-b} is true because \eqref{eqn:sp2id}  holds for every estimation protocol $\{\Set O$, $\Set D\}$ and every state $\V{\rho}_{\mathrm{all}}\in \Set{S}_{\mathrm{sp}}$, and
\eqref{eqn:sp-le-id-c} holds because $\{\Set O^*, \Set D^*\}$ is the optimal solution to $\Ps$-\ref{prob:id-m}.

Moreover, because $\Set S_{\mathrm{sp}} \supset \Set S_{\mathrm{id}}$,
\begin{align}
\underset{\Set O, \Set D}{\mathrm{minimize}}\,e(\Set S_{\mathrm{sp}}, \Set O, \Set D) 
\geq \underset{\Set O, \Set D}{\mathrm{minimize}}\, e(\Set S_{\mathrm{id}}, \Set O, \Set D).
\label{eqn:sp-ge-id}
\end{align}

According to \eqref{eqn:sp-le-id} and \eqref{eqn:sp-ge-id},
\begin{align}
\begin{split}
\underset{\Set O, \Set D}{\mathrm{minimize}}\, e(\Set S_{\mathrm{sp}}, \Set O, \Set D)
&= e(\Set S_{\mathrm{sp}}, \Set O^*, \Set D^*)
\\
&= \underset{\Set O, \Set D}{\mathrm{minimize}}\, e(\Set S_{\mathrm{id}}, \Set O, \Set D)
\end{split}
\label{eqn:sp-eq-id}
\end{align}
The first equality in \eqref{eqn:sp-eq-id} shows that ${\Set O}^*$ is optimal for $\Ps$-\ref{prob:sp-m}.
This completes the proof of \thref{lem:id2sp}.

\section{Proof of \thref{lem:id-sample}}
\label{pf_lem:id-sample}
Because the sample set $\SM$ is drawn completely at random, 
\begin{align}
\sum_{\SM}\bar{f} - \bar{f}_{\Set{M}} =0,\quad \forall \V{\rho}_{\mathrm{all}} \in \Set{S}_{\mathrm{id}}.
\label{eqn:f-f-0}
\end{align}
According to \eqref{eqn:f-f-0}, \eqref{eqn:unbiased_id} is equivalent to
\begin{align}
\sum_{\SM}\mathbb{E}_{R}\big[\check{F}-\bar{f}_{\Set{M}}\big]=0 ,\quad \forall \V{\rho}_{\mathrm{all}}\in{\Set S}_{\mathrm{id}}.\label{eqn:unbiased_id_v2}
\end{align}
\eqref{eqn:unbiased_id_v2} indicates that any estimation protocol $\{\Set O, \Set D\}$ that satisfies \eqref{eqn:unbiased_id2} for all $ {\Set S}_{\mathrm{id}}(\V{f}_{\mathrm{all}})\subset {\Set S}_{\mathrm{id}}$ and all $\SM \subset \Set N$ also satisfies \eqref{eqn:unbiased_id}. 
In the following, we prove the converse statement by contradiction.

Suppose there is an estimation protocol $\{\Set O, \Set D\}$ that satisfies \eqref{eqn:unbiased_id_v2} but does not satisfy \eqref{eqn:unbiased_id2}. 
In this case, there must exist some $\V{\rho}_{\SM}$ such that
\begin{align}
\mathbb{E}_{R}\big[\check{F}-\bar{f}_{\SM}\big| \V{\rho}_{\SM}\big]\neq 0.\label{eqn:biased-id2}
\end{align}
Let $m\in \Set M$ be the index of one of the sampled noisy \ac{GHZ} states, and 
let $l$ be the number of sampled noisy \ac{GHZ} states that have the same state as $\V{\rho}_m$, i.e.,
\begin{align}
\sum_{n \in \SM}  \mathds{1}(\V{\rho}_n = \V{\rho}_m) = l ,
\end{align}
and let $\V{\rho}^{(l)}$ be a state such that \eqref{eqn:biased-id2} holds when $\V{\rho}_{\SM}=\V{\rho}^{(l)}$.
Then because $m\in \Set M$, it is clear that $l\geq 1$.

Next, we construct an $N\times L$ qubits state:
\begin{align}
\V{\rho}_{\mathrm{all}} = \V{\rho}^{(l)} \otimes (\otimes^{N-M} \V{\rho}_m).\label{eqn:rho-N+1}
\end{align}
When sampling the state in \eqref{eqn:rho-N+1}, the sampled set $\SM$ contains at least $l$ number of noisy \ac{GHZ} states whose density matrix is $\V{\rho}_m$.
If this number equals $l$, then from \eqref{eqn:rho-N+1}, the density matrix of the sampled states must be $\V{\rho}^{(l)}$.
Consequently, for \eqref{eqn:unbiased_id_v2} to hold despite \eqref{eqn:biased-id2}, there must exist a sample set $\tSM$
such that 
\begin{align}
\begin{split}
\mathbb{E}_{R}\big[\check{F}-\bar{f}_{\tSM}\big|\V{\rho}_{\tSM} \big]&\neq 0,\quad{\mbox{and}}\\
\sum_{n \in \tSM}  \mathds{1}(\V{\rho}_n = \V{\rho}_m) &=\tilde{l} \geq l+1.
\end{split}\label{eqn:induction_result}
\end{align}

Denote the state $\V{\rho}_{\tSM}$ by $\V{\rho}^{(\tilde{l})}$.
Next, we construct another $N\times L$ qubits state:
\begin{align}
\V{\rho}_{\mathrm{all}} = \V{\rho}^{(\tilde{l})} \otimes (\otimes^{N-M} \V{\rho}_m)\label{eqn:rho-N+2}
\end{align}
and repeat the analysis from \eqref{eqn:rho-N+1} to \eqref{eqn:induction_result}.  
Because $\tilde{l} \geq l+1$, by repeating this analysis at most $M-l$ times, it can be concluded that state 
$\V{\rho}^{(M)}=\otimes^M \V{\rho}_m$ does not satisfy the unbiased constraint, i.e., 
\begin{align}
\mathbb{E}_{R}\big[\check{F}-\bar{f}_{\tSM}\big|\V{\rho}^{(M)} \big]&\neq 0.\label{eqn:biased_id2_v2}
\end{align}
According to \eqref{eqn:biased_id2_v2}, \eqref{eqn:unbiased_id_v2} does not hold for the state $\V{\rho}_{\mathrm{all}}=\otimes^N \V{\rho}_m$.
This contradiction demonstrates that any estimation protocol $\{\Set O, \Set D\}$ satisfying \eqref{eqn:unbiased_id} must also satisfy \eqref{eqn:unbiased_id2}.

Using the strengthened unbiased constraint \eqref{eqn:unbiased_id2}, the estimation error in $\Ps$-\ref{prob:id-m}
can be decomposed as follows:
\begin{subequations}
\begin{align}
&\sum_{\SM}\mathbb{E}_{R}\big[(\check{F}-\bar{f})^2\big]\nonumber \\
&=
\sum_{\SM}\mathbb{E}_{R}\big[(\check{F}-\bar{f}_{\SM} +\bar{f}_{\SM} -\bar{f})^2\big]\label{eqn:decompose-a}\\
&=
\sum_{\SM}\mathbb{E}_{R}\big[(\check{F}-\bar{f}_{\SM})^2\big]
+2 \sum_{\SM}\mathbb{E}_{R}\big[(\check{F}-\bar{f}_{\SM})(\bar{f}_{\SM} -\bar{f})\big]\nonumber\\
 &\hspace{4mm}+\sum_{\SM}\mathbb{E}_{R}\big[(\bar{f}_{\SM} -\bar{f})^2\big]\label{eqn:decompose-b}\\
 &=
\sum_{\SM}\mathbb{E}_{R}\big[(\check{F}-\bar{f}_{\SM})^2\big]
+2 \sum_{\SM}\mathbb{E}_{R}\big[\check{F}-\bar{f}_{\SM}\big| \V{\rho}_{\SM}\big](\bar{f}_{\SM} -\bar{f})\nonumber\\
 &\hspace{4mm}+\sum_{\SM}(\bar{f}_{\SM} -\bar{f})^2\label{eqn:decompose-c}\\
 &=\sum_{\SM}\mathbb{E}_{R}\big[(\check{F}-\bar{f}_{\SM})^2\big]+\sum_{\SM}(\bar{f}_{\SM} -\bar{f})^2,\label{eqn:decompose-d}
\end{align}\label{eqn:decompose}
\end{subequations} 
where \eqref{eqn:decompose-c} holds because with independent noise, 
the average fidelity of sampled noisy \ac{GHZ} states, $\bar{f}_{\SM}$, and that of unsampled noisy \ac{GHZ} states, $\bar{f}$, are deterministic given each sample set $\SM$;
\eqref{eqn:decompose-d} is true due to \eqref{eqn:unbiased_id2}.

\eqref{eqn:decompose} decomposes the estimation error into two parts, i.e., the estimation error \ac{w.r.t.} the average fidelity of the sampled noisy \ac{GHZ} states
\begin{align}
\sum_{\SM}\mathbb{E}_{R}\big[(\check{F}-\bar{f}_{\SM} )^2\big],\label{eqn:error-1}
\end{align}
and the deviation between the average fidelity of sampled and unsampled noisy \ac{GHZ} states
\begin{align}
\sum_{\SM}(\bar{f}_{\SM} -\bar{f})^2.\label{eqn:error-2}
\end{align}
The value of \eqref{eqn:error-2} is determined by the fidelity composition of all noisy \ac{GHZ} states, i.e., 
\begin{align}
\V{f}_{\mathrm{all}} =\{f_n,n\in \Set N\},
\end{align}
and not affected by the estimation protocol $\{\Set O$, $\Set D\}$.
Therefore, for every fidelity composition $\V{f}_{\mathrm{all}}$, minimizing the estimation error
\begin{align}
 \underset{\Set O, \Set D}{\mathrm{minimize}}
 \max_{\V{\rho}_{\mathrm{all}}\in{\Set S}_{\mathrm{id}}(\V{f}_{\mathrm{all}})}\sum_{\SM}\mathbb{E}_{R}\big[(\check{F}-\bar{f})^2\big]
 \label{eqn:minvar_id_f}
 \end{align}
 is equivalent to 
\begin{align}
 \underset{\Set O, \Set D}{\mathrm{minimize}}
 \max_{\V{\rho}_{\mathrm{all}}\in{\Set S}_{\mathrm{id}}(\V{f}_{\mathrm{all}})}\sum_{\SM}\mathbb{E}_{R}\big[(\check{F}-\bar{f}_{\SM})^2\big].
 \end{align} 
Consequently, if a protocol $\{\Set O^*$, $\Set D^*\}$ is optimal in $\Ps$-\ref{prob:id2-m} for all ${\Set S}_{\mathrm{id}}(\V{f}_{\mathrm{all}})\subset{\Set S}_{\mathrm{id}}$, 
it is also the optimal solution to \eqref{eqn:minvar_id_f} for all ${\Set S}_{\mathrm{id}}(\V{f}_{\mathrm{all}})\subset{\Set S}_{\mathrm{id}}$.
Furthermore, because
\begin{align}
{\Set S}_{\mathrm{id}} = \bigcup_{\V{f}_{\mathrm{all}}}{\Set S}_{\mathrm{id}}(\V{f}_{\mathrm{all}}),
\end{align}
protocol $\{\Set O^*$, $\Set D^*\}$ minimizes
\begin{align}
 \max_{\V{\rho}_{\mathrm{all}}\in{\Set S}_{\mathrm{id}}}\sum_{\SM}\mathbb{E}_{R}\big[(\check{F}-\bar{f}\,)^2\big].
 \end{align}
Therefore, $\{\Set O^*$, $\Set D^*\}$ is also an optimal solution to $\Ps$-\ref{prob:id-m}.
This completes the proof of \thref{lem:id-sample}.

\section{Proof of \thref{lem:min}}
\label{pf_lem:min}
Recall the target state $|G^s_{\V{t}}\rangle$. 
Define the following $L$-qubit states:
\begin{align}
\V{\rho}_{n}= \sum_{\tau_{a},\tau_{b}\in \atop \{+,-\}}
\sum_{\V{j}_a,\V{j}_b\in\atop\{0\V{t}_{\sstrike{1}}, 0\tilde{\V{t}}_{\sstrike{1}}\}}
c^{\tau_{a},\tau_{b}}_{n,\V{j}_a,\V{j}_b} |G^{\tau_{a}}_{\V{j}_a} \rangle\langle G^{\tau_{b}}_{\V{j}_b} |, \; n \in {\Set N}\label{eqn:rho_n}
\end{align}
where $\V{t}_{\strike{1}}$ is the substring of $\V{t}$ with the first element removed, $\tilde{\V{t}}_{\sstrike{1}}$ is the bitwise complement of $\V{t}_{\strike{1}}$. In particular,
\begin{align}
\V{t} = 0\V{t}_{\strike{1}} , \quad \tilde{\V{t}} = 1 \tilde{\V{t}}_{\strike{1}}. \label{eqn:Vjt}
\end{align} 
The coefficients $\{c^{\tau_{a},\tau_{b}}_{n,\V{j}_a,\V{j}_b}\}$ are set such that $\V{\rho}_{n}\succcurlyeq \M{0} $,  
$\mathrm{Tr}(\V{\rho}_{n}) =1$, and
\begin{align}
c^{s,s}_{n,\V{t},\V{t}} = f_n.\label{eqn:fni}
\end{align}
Further define the following set of $N\times L$-qubit states:
\begin{align}
\tilde{\Set S}_{\mathrm{id}}(\V{f}_{\mathrm{all}}) = \big\{\V{\rho}_{\mathrm{all}} = \otimes^{n\in\Set N} \V{\rho}_{n}\big\}\label{eqn:setTS}
\end{align}

From the definition above, it is clear that
\begin{align}
{\Set S}_{\mathrm{id}}(\V{f}_{\mathrm{all}}) \supset \tilde{\Set S}_{\mathrm{id}}(\V{f}_{\mathrm{all}}).
\end{align}
Therefore,
\begin{align}
\underset{\Set O, \Set D}{\mathrm{minimize}}\,e({\Set S}_{\mathrm{id}}(\V{f}_{\mathrm{all}}), \Set O, \Set D) 
\geq \underset{\Set O, \Set D}{\mathrm{minimize}}\, e( \tilde{\Set S}_{\mathrm{id}}(\V{f}_{\mathrm{all}}), \Set O, \Set D),
\label{eqn:id-ge-tid}
\end{align}
where the function of the worst-case estimation error, $e({\Set S}, {\Set O}, {\Set D})$, is defined in \eqref{eqn:fSOD}. 

Define a new estimation error minimization problem:
\begin{Problem}[Error minimization for special states]\label{prob:id3}
Same as $\Ps$-\ref{prob:id2-m}, except that the set of states ${\Set S}_{\mathrm{id}}(\V{f}_{\mathrm{all}})$ is replaced by $\tilde{\Set S}_{\mathrm{id}}(\V{f}_{\mathrm{all}})$. 
\end{Problem}
Then \eqref{eqn:id-ge-tid} shows that to prove the lemma, 
it is sufficient to demonstrate that the objective function of $\Ps$-\ref{prob:id3}, when divided by $N \choose M$, is no less than \eqref{eqn:errorbound}, i.e.,
\begin{align*}\sum_{n\in\Set N}\frac{(2f_n+1)(1-f_n)}{2MN}.
\end{align*}

Next, define the following projectors 
\begin{align}
\M{P}_{\strike{1}} & =  |0\rangle\langle\V{t}_{\strike{1}}| + |1 \rangle\langle\tilde{\V{t}}_{\strike{1}}|, \\
\M{P} &= \mathbb{I}_2 \otimes \M{P}_{\strike{1}} \nonumber \\
&= \big(|0\rangle\langle 0 | + |1\rangle\langle 1 |\big) \otimes \big(|0\rangle\langle\V{t}_{\strike{1}}| + |1 \rangle\langle\tilde{\V{t}}_{\strike{1}}|\big)\nonumber  \\
&=  |00\rangle\langle0\V{t}_{\strike{1}}| + |11 \rangle\langle1\tilde{\V{t}}_{\strike{1}}| + |01\rangle\langle0\tilde{\V{t}}_{\strike{1}}| + |10 \rangle\langle1\V{t}_{\strike{1}} |,
\label{eqn:MP}\\
\M{Q} &= \M{P}^\dag \M{P} \nonumber\\
          &= |0\V{t}_{\strike{1}}\rangle\langle0\V{t}_{\strike{1}}| + |1\tilde{\V{t}}_{\strike{1}} \rangle\langle 1\tilde{\V{t}}_{\strike{1}} | 
          + |0\tilde{\V{t}}_{\strike{1}}\rangle\langle0\tilde{\V{t}}_{\strike{1}}| + |1\V{t}_{\strike{1}} \rangle\langle 1\V{t}_{\strike{1}}|.
          \label{eqn:MQ}
\end{align}
From \eqref{eqn:MQ}, for noisy \ac{GHZ} states defined by \eqref{eqn:rho_n}, 
\begin{align}
\V{\rho}_n = \M{Q}\V{\rho}_n \M{Q}^\dag = \M{P}^\dag \M{P} \V{\rho}_n \M{P}^\dag \M{P}.\label{eqn:QrhoQ}
\end{align}
Consequently, in $\Ps$-\ref{prob:id3}, the probability of getting measurment outcome $r$ is given by
\begin{subequations}
\begin{align}
&\mathrm{Tr}\big[\otimes^{n\in \Set M} \V{\rho}_{n}\M{M}_{r}\big] \nonumber \\
&= \mathrm{Tr}\Big[\otimes^{n\in \Set M} \M{P}^\dag \M{P} \V{\rho}_{n} \M{P}^\dag \M{P}\sum_{k} \otimes^{l\in \Set L}\M{M}^{(l)}_{r,n,k} \Big] 
\label{eqn:traceeq-a}\\
&= \mathrm{Tr}\Big[\otimes^{n\in \Set M} \M{P} \V{\rho}_{n} \M{P}^\dag\sum_{k} \M{P} \otimes^{l\in \Set L}\M{M}^{(l)}_{r,n,k}\M{P}^\dag \Big]
\label{eqn:traceeq-b}\\
&= \mathrm{Tr}\Big[\otimes^{n\in \Set M} \hat{\V{\rho}}_{n}  \sum_{k} \M{M}^{(1)}_{r,n,k}\otimes \big(\M{P}_{\strike{1}} \otimes^{l\in \Set L\backslash\{1\}}\M{M}^{(l)}_{r,n,k}\M{P}_{\strike{1}}^\dag\big) \Big]
\label{eqn:traceeq-c}\\
&= \mathrm{Tr}\Big[\otimes^{n\in \Set M} \hat{\V{\rho}}_{n}  \sum_{k}\otimes^{n\in \Set M} \big( \M{M}^{(1)}_{r,n,k}\otimes \M{M}^{(\strike{1})}_{r,n,k}\big)\Big]
\label{eqn:traceeq-d}\\
&= \mathrm{Tr}\big[\otimes^{n\in \Set M} \hat{\V{\rho}}_{n} \hat{\M{M}}_r\big],
\end{align}\label{eqn:traceeq}
\end{subequations}
where 
\begin{align}
\hspace*{-2mm}\hat{\V{\rho}}_{n} &= \M{P} \V{\rho}_{n} \M{P}^\dag \nonumber \\
&= \sum_{\tau_{a},\tau_{b}\in \atop \{+,-\}}
\sum_{\V{j}_a,\V{j}_b\in\atop\{00, 01\}}
c^{\tau_{a},\tau_{b}}_{n,\V{j}_a,\V{j}_b} |G^{\tau_{a}}_{\V{j}_a} \rangle\langle G^{\tau_{b}}_{\V{j}_b} |,\label{eqn:hatrho}\\
\M{M}^{(\strike{1})}_{r,n,k} &= \M{P}_{\strike{1}} \otimes^{l\in \Set L\backslash\{1\}}\M{M}^{(l)}_{r,n,k}\M{P}_{\strike{1}}^\dag,\\
\hat{\M{M}}_r &= \sum_{k}\otimes^{n\in \Set M} \big( \M{M}^{(1)}_{r,n,k}\otimes \M{M}^{(\strike{1})}_{r,n,k}\big),
\label{eqn:hatMr}
 \end{align}
\eqref{eqn:traceeq-a} is obtained according to \eqref{eqn:MXsep} and \eqref{eqn:QrhoQ},
\eqref{eqn:traceeq-b} holds because that $\mathrm{Tr}[\M{A}\M{B}]=\mathrm{Tr}[\M{B}\M{A}]$,
and \eqref{eqn:traceeq-c} is obtained by substituting \eqref{eqn:MP}.

Consider a scenario in which the number of nodes $L=2$, the target state is $|G^s_{00} \rangle$, and the density matrix of all qubits
\begin{align}
\hat{\V{\rho}}_{\mathrm{all}} = \otimes^{n\in\Set N} \hat{\V{\rho}}_{n}.\label{eqn:hatrhoall}
\end{align}
Then we have that
\begin{itemize}
\item According to \eqref{eqn:fni} and \eqref{eqn:hatrho}, $\hat{\V{\rho}}_{\mathrm{n}}$ can be arbitrary two-qubit state with fidelity $f_n$.
\item According to \eqref{eqn:hatMr}, measurement operators $\hat{\M{M}}_r$, $r\in \Set R$ are separable.
\item According to \eqref{eqn:traceeq}, 
applying operation ${\Set O}= \{\M{M}_r\}$ to the state $\V{\rho}_{\mathrm{all}}$ given in \eqref{eqn:setTS}
and applying operation ${\Set O}= \{\hat{\M{M}}_r\}$ to the state $\hat{\V{\rho}}_{\mathrm{all}}$
result in measurement outcomes with the same distribution.
\end{itemize}
Due to the three facts above, $\Ps$-\ref{prob:id3} is equivalent to the following problem.
\begin{Problem}[Error minimization for sampled noisy Bell states]\label{prob:id2-n2}
Same as $\Ps$-\ref{prob:id2-m}, except that the number of nodes $L=2$. 
\end{Problem}
We have proved in \cite[Lem.~B.3]{Ruan:J23} that the objective function of $\Ps$-\ref{prob:id2-n2}, when divided by $N \choose M$,
is the lower bounded by \eqref{eqn:errorbound}.
This proves \thref{lem:min}. 

\section{Proof of \thref{lem:Bloch}}
\label{pf_lem:Bloch}
We first derive the Bloch representation for  $\M{G}^{\pm}_{00\ldots 0}$, then generalize it to all \ac{GHZ} states defined in~\eqref{eqn:gGHZ}.
\begin{align}
\M{G}^{\pm}_{00\ldots 0}&=\frac{1}{2} \big(|0\rangle \langle 0|^{\otimes L}+ |1\rangle \langle 1|^{\otimes L} \pm |0\rangle \langle 1|^{\otimes L} \pm |1\rangle \langle 0|^{\otimes L}\big)
 \nonumber\\
 &= \frac{1}{2} \bigg( \bigg(\frac{\mathbb{I}_2+ \V{\sigma}_z}{2}\bigg)^{\hspace{-1mm}\otimes L}+ \bigg(\frac{\mathbb{I}_2- \V{\sigma_z}}{2}\bigg)^{\hspace{-1mm}\otimes L} \nonumber\\
 &\hspace{5mm} \pm \bigg(\frac{\V{\sigma}_x+  \imath \V{\sigma}_y}{2}\bigg)^{\hspace{-1mm}\otimes L} \pm \bigg(\frac{\V{\sigma}_x-  \imath \V{\sigma}_y}{2}\bigg)^{\hspace{-1mm}\otimes L}\bigg)
\nonumber\\
&= \frac{1}{2^L}\sum_{\substack{ \V{k} \in \{0,1\}^L \\ ||\V{k}||_1\scriptsize \mbox{ even}}}
\Big(\Motimes_{l\in\Set L} \V{\sigma}_{Iz}^{(k_l)} \pm (-1)^{\frac{||\V{k}||_1}{2}}\Motimes_{l\in\Set L} \V{\sigma}_{xy}^{(k_l)}\Big),
 \label{eqn:G000}
\end{align}
where ${\Set L} = \{1,2,\ldots,L\}$, and $\V{\sigma}_{Iz}^{(k_l)}$, $\V{\sigma}_{xy}^{(k_l)}$ are defined in \eqref{eqn:sigmaxyz}.

Denote
\begin{align}
\V{\sigma}_{Ix}^{(j)} = 
\left\{\begin{array}{ll} 
\mathbb{I}_2 & \mbox{ if: } j = 0\\
\V{\sigma}_x & \mbox{ if: } j = 1
\end{array}\right.,
\end{align}
Then with \eqref{eqn:G000}, the Bloch representation of \ac{GHZ} states can be derived as follows.
\begin{align}
\M{G}^{\pm}_{\V{j}}
&= \Motimes_{l\in\Set L} \V{\sigma}_{Ix}^{(j_l)} \M{G}^{\pm}_{00\ldots 0} \Motimes_{l\in\Set L} \V{\sigma}_{Ix}^{(j_l)}\nonumber\\
&= \frac{1}{2^L}\hspace{-1mm}\sum_{\substack{ \V{k} \in \{0,1\}^L \\ ||\V{k}||_1\scriptsize \mbox{ even}}}\hspace{-1mm}
\Big(\Motimes_{l\in\Set L} \V{\sigma}_{Ix}^{(j_l)}\V{\sigma}_{Iz}^{(k_l)}\V{\sigma}_{Ix}^{(j_l)}\nonumber\\ 
&\hspace{24mm}\pm (-1)^{\frac{||\V{k}||_1}{2}}\Motimes_{l\in\Set L} \V{\sigma}_{Ix}^{(j_l)}\V{\sigma}_{xy}^{(k_l)}\V{\sigma}_{Ix}^{(j_l)}\Big)\nonumber\\
& = \frac{1}{2^L}\hspace{-1mm}\sum_{\substack{ \V{k} \in \{0,1\}^L \\ ||\V{k}||_1\scriptsize \mbox{ even}}}\hspace{-1mm}
\Big(\Motimes_{l\in\Set L} (-1)^{j_lk_l}\V{\sigma}_{Iz}^{(k_l)}\nonumber\\ 
&\hspace{24mm}\pm (-1)^{\frac{||\V{k}||_1}{2}}\Motimes_{l\in\Set L}(-1)^{j_lk_l} \V{\sigma}_{xy}^{(k_l)}\Big)\nonumber\\
& = \frac{1}{2^L}\hspace{-1mm}\sum_{\substack{ \V{k} \in \{0,1\}^L \\ ||\V{k}||_1\scriptsize \mbox{ even}}}\hspace{-1mm}
(-1)^{\V{k}\cdot\V{j}}
\Big(\Motimes_{l\in\Set L} \V{\sigma}_{Iz}^{(k_l)} \pm (-1)^{\frac{||\V{k}||_1}{2}}\Motimes_{l\in\Set L} \V{\sigma}_{xy}^{(k_l)}\Big).
\label{eqn:Gj1-L}
\end{align}
With \eqref{eqn:Gj1-L}, the lemma is proved.

\section{Proof of \thref{lem:Paulivalue}}
\label{pf_lem:Paulivalue}
First, we prove two propositions essential for establishing the lemma.
\begin{Prop}
Let $\Set{L}$ be a set with $L$ elements. 
Then $\Set{L}$ has $2^{L-1}$ subsets with an even number of elements and an equal number of subsets with an odd number of elements.
\thlabel{prop:subsetnum}
\end{Prop}
\begin{proof}
Let $l$ be one of the elements in ${\Set L}$. Then, the set of subsets of ${\Set L}$ with an even number of elements is given by
\begin{align}
&\big\{\Set{S}: \Set{S}\subseteq \Set{L} , |\Set S| \mbox{ is even}\big\}\nonumber\\ 
&\quad= 
\big\{\tilde{\Set S}\cup \{l\}: \tilde{\Set S} \subseteq {\Set L}\backslash \{l\}, |\tilde{\Set S}|\mbox{ is odd}\big\}\nonumber\\
&\hspace{11.8mm}\cup
\big\{\tilde{\Set S}: \tilde{\Set S} \subseteq {\Set L}\backslash \{l\}, |\tilde{\Set S}|\mbox{ is even}\big\}.\label{eqn:set_evenodd}
\end{align}

From \eqref{eqn:set_evenodd},  there is a one-to-one correspondence between the subsets of ${\Set L}$ with even number of elements and the subsets of ${\Set L}\backslash \{l\}$. Since set ${\Set L}\backslash \{l\}$ has $2^{L-1}$ number of subsets, it can be obtained that
\begin{align}
\big|\big\{\Set{S}: \Set{S}\subseteq \Set{L} , |\Set S| \mbox{ is even}\big\}\big|=2^{L-1}.
\end{align}

Similarly, it can be proved that
\begin{align}
\big|\big\{\Set{S}: \Set{S}\subseteq \Set{L} , |\Set S| \mbox{ is odd}\big\}\big|=2^{L-1}.
\end{align}
This completes the proof of \thref{prop:subsetnum}.
\end{proof}

The following statement generalizes \thref{prop:subsetnum}.
\begin{Prop}
\thlabel{lem:subsetnum}
Let $\Set{L}$ be a set with $L$ elements. Set $\Set{T}$ is a strict subset of $\Set{L}$ with $T$ number of elements, i.e., 
\begin{align}
\Set{T}\subset\Set{L}, \quad |\Set T| = T \leq L-1.
\end{align}
Then the numbers of sets $\Set S$ satisfying
\begin{align}
{\Set T}\subseteq {\Set S} \subseteq{\Set L}, \;  |\Set S| \mbox{ is even},
\end{align}
or
\begin{align}
{\Set T}\subseteq {\Set S} \subseteq{\Set L}, \;  |\Set S| \mbox{ is odd}
\end{align}
are $2^{L-T-1}$ in both cases.
\end{Prop}
\begin{proof}
If $|\Set T |$ is even, then
\begin{align} 
&\big\{{\Set S}: {\Set T}\subseteq {\Set S} \subseteq{\Set L},  |\Set S| \mbox{ is even}\big\}  \nonumber\\
&\qquad=\big\{\tilde{\Set S}\cup {\Set T}:  \tilde{\Set S}\subseteq{\Set L}\backslash {\Set T},   |\tilde{\Set S}| \mbox{ is even} \big\}.
\end{align}
In this case, from \thref{prop:subsetnum},
\begin{align} 
&\big|\big\{{\Set S}: {\Set T}\subseteq {\Set S} \subseteq{\Set L},  |\Set S| \mbox{ is even}\big\}\big| \nonumber\\
&\qquad= 
\big|\big\{\tilde{\Set S}:  \tilde{\Set S}\subseteq{\Set L}\backslash {\Set T},   |\tilde{\Set S}| \mbox{ is even} \big\}\big|
= 2^{L-T-1}.
\end{align}
Similarly, if $|\Set T |$ is odd, then
\begin{align} 
&\big\{{\Set S}: {\Set T}\subseteq {\Set S} \subseteq{\Set L},  |\Set S| \mbox{ is even}\big\} \nonumber\\
&\qquad= 
\big\{\tilde{\Set S}\cup {\Set T}:  \tilde{\Set S}\subseteq{\Set L}\backslash {\Set T},   |\tilde{\Set S}| \mbox{ is odd} \big\}.
\end{align}
Therefore, from \thref{prop:subsetnum},
\begin{align} 
&\big|\big\{{\Set S}: {\Set T}\subseteq {\Set S} \subseteq{\Set L},  |\Set S| \mbox{ is even}\big\}\big| \nonumber\\
&\qquad= \big|\big\{\tilde{\Set S}:  \tilde{\Set S}\subseteq{\Set L}\backslash {\Set T},   |\tilde{\Set S}| \mbox{ is odd} \big\}\big|
= 2^{L-T-1}.
\end{align}
In both cases, 
\begin{align} 
\big|\big\{{\Set S}: {\Set T}\subseteq {\Set S} \subseteq{\Set L},  |\Set S| \mbox{ is even}\big\}\big|  = 2^{L-T-1}.
\end{align}
By repeating the above analysis, one can also conclude that
\begin{align} 
&\big|\big\{{\Set S}: {\Set T}\subseteq {\Set S} \subseteq{\Set L},  |\Set S| \mbox{ is odd}\big\}\big|
= 2^{L-T-1}.
\end{align}
This completes the proof of \thref{lem:subsetnum}.
\end{proof}

According to \thref{prop:subsetnum}, the size of $\Set{L}_{\mathrm{even}}$ is $2^{L-1}$.
Therefore,  \eqref{eqn:obv_sumspm_w} and \eqref{eqn:obv_sums0_w} hold true when the following equations are satisfied.
\begin{align}
\pm (-1)^{\frac{||\V{k}||_1}{2} + \V{k}\cdot\V{j}}\mathrm{Tr}\big[\V{\sigma}_{xy}^{(\V{k})} \M{G}^{\pm}_{\V{j}}\big] =1,\label{eqn:obv_sumspm}\\
\forall\; ||\V{k}||_1 \mbox{ even}, \V{j} \in {\Set J},\mbox{ and}\nonumber\\
\sum_{\scriptsize \V{k} \in \{0,1\}^L \atop ||\V{k}||_1 \mbox{ even}}(-1)^{\frac{||\V{k}||_1}{2} + \V{k}\cdot\V{j}}\mathrm{Tr}\big[\V{\sigma}_{xy}^{(\V{k})} \M{G}^{\pm}_{\TVj}\big] =0,\label{eqn:obv_sums0}\\
\forall\; \TVj \in {\Set J}, \TVj \neq \V{j}.\nonumber
\end{align}
The following analysis proves \eqref{eqn:obv_sumspm} and \eqref{eqn:obv_sums0} sequentially.

From \eqref{eqn:gGHZBloch} and \eqref{eqn:sigmak-2},  the value of $\V{\sigma}_{xy}^{(\V{k})}$ operated on  $|G^{\pm}_{\V{j}}\rangle$ can be obtained as follows.
\begin{align}
&\mathrm{Tr}\big[\V{\sigma}_{xy}^{(\V{k})} \M{G}^{\pm}_{\V{j}}\big] \nonumber\\
&= \mathrm{Tr}\bigg[
\frac{1}{2^L}\sum_{\substack{ \V{k} \in \{0,1\}^L \\ ||\V{k}||_1\scriptsize \mbox{ even}}}
(-1)^{\V{k}\cdot\V{j}}
\Big(\Motimes_{l\in\Set L} \V{\sigma}_{xy}^{(k_l)} \V{\sigma}_{Iz}^{(k_l)} \nonumber\\
&
\hspace{20.1mm}\pm (-1)^{\frac{||\V{k}||_1}{2}}\Motimes_{l\in\Set L} \V{\sigma}_{xy}^{(k_l)} \V{\sigma}_{xy}^{(k_l)}\Big)
\bigg]\nonumber\\
&=\mathrm{Tr}\bigg[\pm \frac{1}{2^L} (-1)^{\frac{||\V{k}||_1}{2} + \V{k}\cdot\V{j}}\Motimes_{l\in\Set L} \V{\sigma}_{xy}^{(k_l)} \V{\sigma}_{xy}^{(k_l)}\Big)\bigg]\nonumber\\
&= \pm (-1)^{\frac{||\V{k}||_1}{2} + \V{k}\cdot\V{j}} \mathrm{Tr}\bigg(\frac{1}{2^L}\mathbb{I}_{2^{L}}\bigg)\nonumber\\
& = \pm (-1)^{\frac{||\V{k}||_1}{2} + \V{k}\cdot\V{j}} \in \{-1,1\}.\label{eqn:obvS0G}
\end{align}
From \eqref{eqn:obvS0G}, it is straightforward that
\begin{align}
\pm (-1)^{\frac{||\V{k}||_1}{2} + \V{k}\cdot\V{j}}\mathrm{Tr}\big[\V{\sigma}_{xy}^{(\V{k})} \M{G}^{\pm}_{\V{j}}\big] &= 1.
\end{align}
Therefore, \eqref{eqn:obv_sumspm} is proved. It can also be derived from \eqref{eqn:obvS0G}  that
\begin{align}
\sum_{\substack{ \V{k} \in \{0,1\}^L \\ ||\V{k}||_1\scriptsize \mbox{ even}}}\hspace{-1mm}
(-1)^{\frac{||\V{k}||_1}{2} + \V{k}\cdot\V{j}}\mathrm{Tr}\big[\V{\sigma}_{xy}^{(\V{k})} \M{G}^{\pm}_{\TVj}\big] & = 
\pm\hspace{-1mm}\sum_{\substack{ \V{k} \in \{0,1\}^L \\ ||\V{k}||_1\scriptsize \mbox{ even}}}\hspace{-1mm}(-1)^{\V{k}\cdot\V{j} +\V{k}\cdot\TVj}.
\label{eqn:obv_sums2}
\end{align}
From \eqref{eqn:obv_sums2}, equation \eqref{eqn:obv_sums0} holds true if 
\begin{align}
\sum_{\substack{ \V{k} \in \{0,1\}^L \\ ||\V{k}||_1\scriptsize \mbox{ even}}}(-1)^{\V{k}\cdot\V{j} +\V{k}\cdot\TVj}=0,\quad\forall\; \TVj \in {\Set J}, \TVj \neq \V{j}.\label{eqn:obv_sumseq0}
\end{align}
Denote $\Set{G}$ as the set of indexes $l$ with which $j_l$ and $\tilde{j}_l$ do not equal, i.e., 
\begin{align}
\Set{G}  =  \{l: l\in {\Set L}, j_l \neq \tilde{j}_l\},
\end{align}
and denote $G=|\Set{G}|$. Noticing that as $ \TVj \neq \V{j}$ and $\tilde{j}_1 = j_1 =0$,
\begin{align}
1 \leq G \leq L-1.
\end{align}
Further denote $\Set S$ as the support of $\V{k}$, i.e.,
\begin{align}
\Set S  =  \{l: l\in {\Set L}, k_l \neq 0\},
\end{align}
and $\Set T  =   \Set {S} \cap \Set{G}$, then it can be obtained that
\begin{subequations}
\begin{align}
&\sum_{\substack{ \V{k} \in \{0,1\}^L \\ ||\V{k}||_1\scriptsize \mbox{ even}}}(-1)^{\V{k}\cdot\V{j} +\V{k}\cdot\TVj}\nonumber\\
&=
\sum_{\substack{ \V{k} \in \{0,1\}^L \\ ||\V{k}||_1\scriptsize \mbox{ even}}}(-1)^{|\Set {S} \cap \Set{G}|}
=
\sum_{\substack{ \Set{S} \in \Set{L} \\ |\Set S| \scriptsize \mbox{ even}}}(-1)^{|\Set {S} \cap \Set{G}|}\nonumber\\
&=\sum_{\Set{T} \subseteq \Set{G}} (-1)^{|\Set {T}|} \big|\big\{\Set{S}: \Set{T}\subseteq \Set{S} \subseteq (\Set{L}\backslash\Set{G})\cup\Set{T},|\Set S| \mbox{ is even} \big\}\big|\nonumber\\%\label{eqn:Snuma}\\
&= \sum_{\Set{T} \subseteq \Set{G}} (-1)^{|\Set {T}|} 2^{L-G+|\Set {T}|-|\Set {T}|-1}\label{eqn:Snumb}\\
&=2^{L-G-1}\Big(\big|\big\{{\Set T}: \Set{T} \subseteq \Set{G}, |\Set {T}| \mbox{ is even}  \big\} \big| \nonumber\\
&\hspace{15.5mm}-\big|\big\{\Set T: \Set{T} \subseteq \Set{G}, |\Set {T}| \mbox{ is odd}  \big\} \big|\Big)\nonumber\\
&=2^{L-G-1}(2^{G-1}-2^{G-1}) = 0,\label{eqn:Snumc}
\end{align}\label{eqn:Snum}
\end{subequations}
where %\eqref{eqn:Snuma}  is true because that $(-1)^{\V{k}\cdot\V{j} +\V{k}\cdot\TVj} = -1^{|\Set {S} \cap \Set{D}|}$, and
\eqref{eqn:Snumb}, \eqref{eqn:Snumc} are obtained by applying \thref{lem:subsetnum}. With \eqref{eqn:Snum}, \eqref{eqn:obv_sumseq0} is proved. This completes the proof of \thref{lem:Paulivalue}.

\section{Proof of \thref{lem:unbiased}}
\label{pf_lem:unbiased}
Denote
\begin{align}
f^{\pm}_{\V{j}} = \frac{1}{M}\sum_{n\in \Set M} \mathrm{Tr}\big[\V{\rho}_n \M{G}^{\pm}_{\V{j}}\big], \quad \V{j} \in \Set{J}.
\end{align}
Without loss of generality, assume that the target state is $|G^+_{\V{t}}\rangle$.  
In this case,
\begin{align}
\bar{f}_{\SM} &= f^{+}_{\V{t}}.\label{eqn:fstar}
\end{align}
Furthermore, because the \ac{GHZ} states in set ${\Set G}_L$ form a basis of $\Set{H}^{2^L}$,
\begin{align}
\sum_{\V{j} \in {\Set J} \atop \tau \in \{+,-\} }\V{G}^{\tau}_{\V{j}} & = \mathbb{I}_{2^L}, \label{eqn:sumG}\\
\sum_{\V{j}\in {\Set J} \atop \tau \in \{+ , -\}} f^{\tau}_{\V{j}}  &=  \frac{1}{M}\sum_{n\in \SM}\mathrm{Tr}(\V{\rho}_n) = 1.\label{eqn:sumf1}
\end{align}

Consider the measurement operation in Protocol~\ref{alg:fidelityest}.
According to \eqref{eqn:r_n-z} and \eqref{eqn:sumG}, when $A_n=0$, the operators are given by
\begin{subequations}
\begin{align}
\M{M}_{0} &= |\V{t}\rangle\langle\V{t}| + |\tilde{\V t}\rangle\langle\tilde{\V t}| =  \V{G}^{+}_{\V{t}}  + \V{G}^{-}_{\V{t}},
\label{eqn:M0wA=0}\\
\M{M}_{1} &= \mathbb{I}_{2^L} - \M{M}^{(0)}_{0} \hspace{2mm}=\hspace{-1.8mm}
\sum_{\V{j} \in {\Set J} \backslash \{\V{t}\} \atop \tau \in \{+,-\} }\V{G}^{\tau}_{\V{j}} 
\label{eqn:M1wA=0}.
\end{align}\label{eqn:MwA=0}
\end{subequations}
According to \eqref{eqn:An} and \eqref{eqn:MwA=0}, 
\begin{subequations}
\begin{align}
&\mathbb{E}\bigg[\frac{\sum_{n\in \Set M} R_n \indicator(A_n = 0)}{M}\bigg|\bar{f}_{\SM}\bigg] \nonumber\\
&=\frac{1}{M} \sum_{n\in \Set M}\Pr\big[A_n = 0 \big]\mathbb{E}\big[R_n\big| A_n = 0,\bar{f}_{\SM} \big]\\
&=\frac{1}{3M} \sum_{n\in \Set M} \mathrm{Tr}\big[\V{\rho}_n \M{M}_{1} \big]\\
&= \frac{1}{3}\sum_{\V{j} \in {\Set J} \backslash \{\V{t}\} \atop \tau \in \{+,-\} }f^{\tau}_{\V{j}}.
\end{align}
\label{eqn:ER0}
\end{subequations}

According to \eqref{eqn:sigmak} and \eqref{eqn:r_n-xy}, when $A_n=1$ and $\V{B}_n = \V{b}$, the operators are given by
\begin{align}
\M{M}^{(\V{b})}_{r} &=  \frac{1}{2}\big(\mathbb{I}_{2^L} - (-1)^{\frac{||\V{b}||_1} 2 + \V{b} \cdot \V{t}+r} \V{\sigma}^{(\V{b})}_{xy}\big), \quad r\in\{0,1\}
\label{eqn:MwA=1},
\end{align}
In this case,
\begin{subequations}
\begin{align}
&\hspace{-4mm}\mathbb{E}\bigg[\frac{\sum_{n\in \Set M} R_n \indicator(A_n =1)}{M}\bigg]\nonumber \\
&= \frac{2}{3M} \sum_{n\in \Set M}\sum_{\V{k}\in \{0,1\}^L\atop ||\V{k}||_1\, \mathrm{ even}}\Pr\big[\V{B}_n = \V{b} \big]\mathbb{E}\big[R_n\big| A_n = 1, \V{B}_n = \V{b} \big]\nonumber\\
&= \frac{2}{3M2^{L}} \sum_{n\in \Set M}\sum_{\V{b}\in \{0,1\}^L\atop ||\V{b}||_1\, \mathrm{ even}} 
1- (-1)^{\frac{||\V{b}||_1} 2 + \V{b} \cdot \V{t}}
\mathrm{Tr}\big[\V{\rho}_n\V{\sigma}^{(\V{b})}_{xy}  \big] \label{eqn:ER1_a}\\
&= \frac{2}{3M2^{L}} \sum_{n\in \Set M}\sum_{\V{b}\in \{0,1\}^L\atop ||\V{b}||_1\, \mathrm{ even}} 
1- (-1)^{\frac{||\V{b}||_1} 2 + \V{b} \cdot \V{t}} \nonumber\\
&\hspace{16.5mm}\mathrm{Tr}\Big[ \sum_{\V{j}\in\Set J \atop \tau \in \{+,-\} }
\mathrm{Tr}[\V{\rho}_n\M{G}^\tau_{\V{j}}]\M{G}^\tau_{\V{j}}\V{\sigma}^{(\V{b})}_{xy} \Big]\label{eqn:ER1_b}\\
&=\frac{2}{3M2^{L}} \sum_{n\in \Set M}2^{L-1} -  \sum_{\V{j}\in\Set J \atop \tau \in \{+,-\} }
\mathrm{Tr}[\V{\rho}_n\M{G}^\tau_{\V{j}}] \nonumber\\
&\hspace{15mm}\sum_{\V{b}\in \{0,1\}^L\atop ||\V{k}||_1\, \mathrm{ even}}(-1)^{\frac{||\V{b}||_1} 2 + \V{b} \cdot \V{t}}
\mathrm{Tr}[\V{\sigma}^{(\V{k})}_{xy}\M{G}^\tau_{\V{j}}]\label{eqn:ER1_c}\\
&=\frac{1}{3M} \sum_{n\in \Set M} 1 - \mathrm{Tr}[\V{\rho}_n\M{G}^+_{\V{t}}] + \mathrm{Tr}[\V{\rho}_n\M{G}^-_{\V{t}}]\label{eqn:ER1_d}\\
&=\frac{1}{3M} \sum_{n\in \Set M}\Big( 2\mathrm{Tr}[\V{\rho}_n\M{G}^-_{\V{t}}] + 
\sum_{\V{j} \in {\Set J} \backslash \{\V{t}\}  \atop \tau \in \{+,-\}  }\mathrm{Tr}\big[\V{\rho}_n\M{G}^\tau_{\V{j}}\big]\Big)\label{eqn:ER1_e}\\
&=\frac{2}{3}f^{-}_{\V{t}} + \frac{1}{3} \sum_{\V{j} \in {\Set J} \backslash \{\V{t}\}  \atop \tau \in \{+,-\}  }f^{\tau}_{\V{j}},
\end{align}\label{eqn:ER1}
\end{subequations}
%Denote $M(\V{\sigma}(\Set{S}_a),\V{\rho}_n)\in\{1,-1\}$ as the random measurement outcome when measuring Pauli observable $\V{\sigma}(\Set{S}_a)$ on $\V{\rho}_n$. 
where \eqref{eqn:ER1_a} is true because according to \thref{prop:subsetnum}, 
\begin{align}
\Pr\big[\V{B}_n = \V{b} \big] = \frac{1}{2^{L-1}},
\end{align}
and according to \eqref{eqn:MwA=1},
\begin{subequations}
\begin{align}
&\mathbb{E}\big[R_n\big| A_n= 1, \V{B}_n = \V{b} \big] \nonumber\\
&= \mathrm{Tr}\big[\V{\rho}_n\M{M}^{(\V{b})}_1  \big] \\
&=\frac{1}{2}\Big(1 - (-1)^{\frac{||\V{b}||_1} 2 + \V{b} \cdot \V{t}}
\mathrm{Tr}\big[\V{\rho}_n\V{\sigma}^{(\V{b})}_{xy}  \big]  \big)\Big),
\end{align}
\end{subequations}
\eqref{eqn:ER1_b} and \eqref{eqn:ER1_e} are true because that the \ac{GHZ} states in set $\Set{G}_L$ form a basis of $\Set{H}^{2^L}$, hence
\begin{align}
\V{\rho}_n &= \sum_{\V{j}\in\Set J \atop \tau \in \{+,-\} } \mathrm{Tr}(\V{\rho}_n\M{G}^{\tau}_{\V{j}})\M{G}^{\tau}_{\V{j}},\quad\mbox{and}\\
1 &= \sum_{\V{j}\in\Set J \atop \tau \in \{+,-\} } \mathrm{Tr}(\V{\rho}_n\M{G}^\tau_{\V{j}}),
\end{align}
\eqref{eqn:ER1_c} is derived by taking the scalars $\mathrm{Tr}[\V{\rho}_n\M{G}^{\pm}_{\V{j}}]$ out of the trace operators, 
and \eqref{eqn:ER1_d} is derived by applying \thref{lem:Paulivalue}.

Substituting \eqref{eqn:fstar}, \eqref{eqn:sumf1}, \eqref{eqn:ER0} and \eqref{eqn:ER1} into \eqref{eqn:checkF}, one can obtain that
\begin{align}
\begin{split}
\mathbb{E}\big[\check{F}\big] & = 1 - \frac{3}{2}\mathbb{E}\bigg[\frac{\sum_{n\in \Set M}R_n}{M}\bigg]\\
& = 1- \frac{3}{2}\bigg(\mathbb{E}\bigg[\frac{\sum_{n\in \Set M} R_n \indicator(A_n = 0)}{M}\bigg]  \\
&\hspace{11.6mm}+ \mathbb{E}\bigg[\frac{\sum_{n\in \Set M} R_n \indicator(A_n > 0)}{M}\bigg]\bigg)\\
& = 1 - f^{-}_{\V{t}} - \sum_{\V{j} \in {\Set J} \backslash \{\V{t}\} \atop \tau \in \{+,-\}  }f^{\tau}_{\V{j}} \\
& = f^+_{\V{t}} = \bar{f}_{\SM}.
\end{split}\label{eqn:Funbias}
\end{align}
\eqref{eqn:Funbias} shows that the proposed estimator is unbiased. This completes the proof of \thref{lem:unbiased}.

\section{Proof of \thref{lem:opt-id}}
\label{pf_lem:opt-id}
By repeating the derivation of \eqref{eqn:ER0}, \eqref{eqn:ER1} and \eqref{eqn:Funbias} for the $n$-th \ac{GHZ} state, the following expression is obtained
\begin{align}
\mathrm{Pr}(R_n =1) = \frac{2}{3}(1-f_n).
\end{align}
Hence, the variance of $R_n$ is given by
\begin{align}
\mathbb{V}\big[R_n\big]&= \mathrm{Pr}(R_n =1) (1-\mathrm{Pr}(R_n =1) )\nonumber\\
 &= \frac{(2f_n+1)(2-2f_n)}{9}\label{eqn:VRn}
\end{align}

In the case of independent noise, the measurement outcomes $R_n$ corresponding to different indices $n\in \SM$  are independent.
Therefore, according to \eqref{eqn:checkF}
\begin{align}
\mathbb{V}\big[\check{F}\big| \SM\big]
&= \frac{9}{4M^2}\sum_{n\in\SM} \mathbb{V}\big[R_n\big]\nonumber\\
&=\sum_{n\in\SM}\frac{(2f_n+1)(1-f_n)}{2M^2}.\label{eqn:Fvar-M}
\end{align}

According to Lemma~\ref{lem:unbiased}, $\check{F}$ satisfies \eqref{eqn:unbiased_id2}, therefore
\begin{align}
\mathbb{E}_{R}\big[(\check{F}-\bar{F}_{\SM})^2\big| \SM \big]= \mathbb{V}\big[\check{F}\big| \SM\big].
\label{eqn:E2V-F}
\end{align}
Since the sample set $\SM$ is selected completely at random, \eqref{eqn:Fvar-M} and \eqref{eqn:E2V-F} indicate that
for all ${\Set S}_{\mathrm{id}}(\V{f}_{\mathrm{all}})\subset{\Set S}_{\mathrm{id}}$,
\begin{align}
&{N\choose M}^{-1} \sum_{\SM \subset \Set N}\mathbb{E}_{R}\big[(\check{F}-\bar{F}_{\SM})^2 \big]\nonumber\\ 
& = {N\choose M}^{-1} \sum_{\SM \subset \Set N}\sum_{n\in\SM}\frac{(2f_n+1)(1-f_n)}{2M^2}\nonumber\\
& = {N\choose M}^{-1} \sum_{n\in\Set N}{N-1\choose M-1}\frac{(2f_n+1)(1-f_n)}{2M^2}\nonumber\\
&=\sum_{n\in\Set N}\frac{(2f_n+1)(1-f_n)}{2MN}.\label{eqn:error-upperbound}
\end{align}

According to \thref{lem:min}, the estimation error of $\Ps$-\ref{prob:id2-m} is no less than \eqref{eqn:error-upperbound}.
This demonstrates that Protocol~\ref{alg:fidelityest} is optimal in $\Ps$-\ref{prob:id2-m} for all fidelity compositions ${\Set S}_{\mathrm{id}}(\V{f}_{\mathrm{all}})\subset{\Set S}_{\mathrm{id}}$.
Thus, the proof of \thref{lem:opt-id} is complete.

\section{Proof of \thref{thm:opt}}
\label{pf_thm:opt}
Let $\Set O^*$ and $\Set D^*$ denote the measurement operation and estimator of Protocol~\ref{alg:fidelityest}, respectively.
Define a composite measurement operation $\hat{\Set O}^* = \Set O^* \circ \Set T$.
Then, according to \thref{lem:opt-id} and \thref{thm:gen-opt-id}, the estimation protocol $\{\hat{\Set O}^*, \Set D^* \}$ is optimal for $\Ps$-\ref{prob:1-m}.

When $A_n=0$, it follows from \eqref{eqn:Txy} and \eqref{eqn:MwA=0} that the rotations in operation $\Set T$  do not alter  the  
measurement operators $\M{M}_{r}$. 
Explicitly,
\begin{subequations}
\begin{align}
&(\M{T}^{(\V{k})})^\dag \M{M}_{0} \M{T}^{(\V{k})} \nonumber\\ 
&=  (\V{\sigma}^{(\V{k})}_{xy})^\dag (\V{G}^{+}_{\V{t}}  + \V{G}^{-}_{\V{t}})\V{\sigma}^{(\V{k})}_{xy} \label{eqn:TMrA0-a}\\
&= \V{G}^{+}_{\V{t}}  + \V{G}^{-}_{\V{t}} \label{eqn:TMrA0-b}\\
&= \M{M}_{0},
\end{align}\label{eqn:TM0rA0}
\end{subequations}
and
\begin{subequations}
\begin{align}
&(\M{T}^{(\V{k})})^\dag \M{M}_{1} \M{T}^{(\V{k})} \nonumber\\ 
&=  (\V{\sigma}^{(\V{k})}_{xy})^\dag (\mathbb{I}_{2^L} - \M{M}_{0})\V{\sigma}^{(\V{k})}_{xy} \\
&=  (\V{\sigma}^{(\V{k})}_{xy})^\dag\V{\sigma}^{(\V{k})}_{xy}  - \M{M}_{0} \label{eqn:TM1rA0-b}\\
&= \M{M}_{1},
\label{eqn:TM1rA0-c}
\end{align}\label{eqn:TM1rA0}
\end{subequations}
where $^\dag$ is the Hermitian transpose, \eqref{eqn:TMrA0-b} is obtained by substituting \eqref{eqn:TkGj1j2} into \eqref{eqn:TMrA0-a},
\eqref{eqn:TM1rA0-c} holds because that $\V{\sigma}^\dag_{x}\V{\sigma}_{x} = \V{\sigma}^\dag_{y}\V{\sigma}_{y} = \mathbb{I}_2$.

When $A_n=1$,   it follows from \eqref{eqn:Txy} and \eqref{eqn:MwA=1} that the rotations in operation $\Set T$  do not alter the  
measurement operators $\M{M}^{(\V{b})}_{r}$. 
Explicitly, for any strings $\V{k}$ and $\V{b}$ satisfying $||\V{k}||_1 \in \{0,2\}$ and $||\V{b}||_1$ being even,
\begin{subequations}
\begin{align}
&(\M{T}^{(\V{k})})^\dag \M{M}^{(\V{b})}_{r} \M{T}^{(\V{k})} \nonumber\\
&=\frac{1}{2}(\V{\sigma}^{(\V{k})}_{xy})^\dag\Big(\mathbb{I}_{2^L} - (-1)^{\frac{||\V{b}||_1} 2 + \V{b} \cdot \V{t} +r} \V{\sigma}^{(\V{b})}_{xy}\Big)\V{\sigma}^{(\V{k})}_{xy} 
\label{eqn:TMbrA1-a}\\
&=\frac{1}{2}\Big(\mathbb{I}_{2^L} - (-1)^{\frac{||\V{b}||_1} 2 + \V{b} \cdot \V{t} +r} 
\Motimes_{l\in \Set L}\V{\sigma}^{(k_l)}_{xy}\V{\sigma}^{(b_l)}_{xy}\V{\sigma}^{(k_l)}_{xy}\Big)
\label{eqn:TMbrA1-b}\\
&=\frac{1}{2}\Big(\mathbb{I}_{2^L} - (-1)^{\frac{||\V{b}||_1} 2 + \V{b} \cdot \V{t} +r + \V{b} \oplus \V{k}} 
\Motimes_{l\in \Set L}\V{\sigma}^{(b_l)}_{xy}\Big)
\label{eqn:TMbrA1-c}\\
&=\frac{1}{2}\Big(\mathbb{I}_{2^L} - (-1)^{\frac{||\V{b}||_1} 2 + \V{b} \cdot \V{t} +r } 
\V{\sigma}^{(\V{b})}_{xy}\Big)
\label{eqn:TMbrA1-d}\\
&=\M{M}^{(\V{b})}_{r},
\end{align}\label{eqn:TMbrA1}
\end{subequations}
where $\oplus$ represents the XOR operation,
\eqref{eqn:TMbrA1-b} holds because that $\V{\sigma}^\dag_{x}\V{\sigma}_{x} = \V{\sigma}^\dag_{y}\V{\sigma}_{y} = \mathbb{I}_2$,
\eqref{eqn:TMbrA1-c} holds because that for $u,v\in\{x,y\}$
\begin{align}
\V{\sigma}_{u}\V{\sigma}_{v}\V{\sigma}_{u} =
\left\{\begin{array}{ll} \V{\sigma}_{v} &\mbox{if } u=v \\
-\V{\sigma}_{v} &\mbox{if } u\neq v
\end{array}\right.,
\end{align}
and \eqref{eqn:TMbrA1-d} holds because that both $||\V{k}||_1$ and $||\V{b}||_1$ are even, hence $\V{b} \oplus \V{k}$ is even as well.

According to \eqref{eqn:TM0rA0},  \eqref{eqn:TM1rA0} and \eqref{eqn:TMbrA1}, the operation $\Set{T}$ does not change the operators of $\Set O^*$.
Consequently, the operations $\hat{\Set O}^*$ and $\Set O^*$ are equivalent, i.e., 
\begin{align}
\hat{\Set O}^* = \Set O^* \circ \Set T =\Set O^*.\label{eqn:OT=O_ap}
\end{align}
Thus, by \eqref{eqn:OT=O_ap},  Protocol~ref{alg:fidelityest} is the optimal solution of $\Ps$-\ref{prob:1-m}. This completes the proof of \thref{thm:opt}.


%% Here is the endmatter stuff: Supplementary Info, etc.
%% Use \item's to separate, default label is "Acknowledgements"

%\begin{addendum}
% \item Put acknowledgements here.
% \item[Competing Interests] The authors declare that they have novcompeting financial interests.
% \item[Correspondence] Correspondence and requests for materials
%should be addressed to A.B.C.~(email: myaddress@nowhere.edu).
%\end{addendum}

%%
%% TABLES
%%
%% If there are any tables, put them here.
%%

\end{document}